\documentclass[11pt,letterpaper]{article}

\usepackage{fullpage}
\usepackage{graphicx}
\usepackage{amsmath,amssymb,amsthm}
\usepackage{hyperref}
\usepackage{enumitem}

\theoremstyle{plain}
\newtheorem{theorem}{Theorem}[section]
\newtheorem{lemma}[theorem]{Lemma}

\newtheorem{corollary}[theorem]{Corollary}

\theoremstyle{definition}
\newtheorem{definition}{Definition}[section]
\newtheorem{remark}[theorem]{Remark}

\newcommand{\R}{\mathbb{R}}
\newcommand{\C}{\mathbb{C}}
\newcommand{\Z}{\mathbb{Z}}
\newcommand{\F}{\mathbb{F}}
\newcommand{\cL}{\mathcal{L}}
\newcommand{\cC}{\mathcal{C}}
\renewcommand{\epsilon}{\eps}
\newcommand{\eps}{\varepsilon}
\newcommand{\bS}{\mathbb{S}}
\newcommand{\bcube}{\{-1,1\}}
\DeclareMathOperator{\sign}{sign}

\DeclareMathOperator{\poly}{poly}
\DeclareMathOperator{\polylog}{polylog}
\DeclareMathOperator{\E}{\mathbb{E}}
\DeclareMathOperator{\Var}{Var}

\DeclareMathOperator{\Unif}{Unif}

\DeclareMathOperator{\cost}{cost}
\DeclareMathOperator{\dist}{dist}
\DeclareMathOperator{\proj}{proj}

\newcommand{\iid}{i.i.d.\@ }
\renewcommand{\Re}{\operatorname{Re}}

\newcommand{\eqdist}{\overset{\mathrm{dist}}{=\joinrel=}}
\DeclareMathOperator{\rank}{rank}

\title{Tight Bounds for the Subspace Sketch Problem with Applications\footnotetext{Yi Li was supported in part by a Singapore Ministry of Education (AcRF) Tier 2 grant MOE2018-T2-1-013. Ruosong Wang and David P. Woodruff were supported in part by an Office of Naval Research (ONR) grant N00014-18-1-2562, as well as the Simons Institute for the Theory of Computing where part of this work was done.}}
\author{Yi Li\\
			Nanyang Technological University\\
			\texttt{yili@ntu.edu.sg}
			\and
			Ruosong Wang\\
			Carnegie Mellon University\\
			\texttt{ruosongw@andrew.cmu.edu}
			\and
			David P.\@ Woodruff\\
			Carnegie Mellon University\\
			\texttt{dwoodruf@cs.cmu.edu}}
\date{}

\begin{document}

\clearpage\maketitle
\thispagestyle{empty}
\setcounter{page}{0}

\begin{abstract}
In the {\it subspace sketch problem} one is given an $n \times d$ matrix $A$ with $O(\log(nd))$ bit entries, and would like to compress it in an arbitrary way to build a small space data structure $Q_p$, so that for any given $x \in \mathbb{R}^d$, with probability at least $2/3$, one has $Q_p(x) = (1 \pm \eps) \|Ax\|_p$, where $p \geq 0$ and the randomness is over the construction of $Q_p$. The central question is: 
\begin{center}
{\it How many bits are necessary to store $Q_p$?}
\end{center}
This problem has applications to the communication of approximating the number of non-zeros in a matrix product, the size of coresets in projective clustering, the memory of streaming algorithms for regression in the row-update model, and embedding subspaces of $L_p$ in functional analysis. A major open question is the dependence on the approximation factor $\epsilon$. 

We show if $p \geq 0$ is {\it not a positive even integer} and $d = \Omega(\log(1/\eps))$, then $\widetilde{\Omega}(\eps^{-2} \cdot d)$ bits are necessary. On the other hand, if $p$ is a positive even integer, then there is an upper bound of $O(d^p \log(nd))$ bits independent of $\eps$. Our results are optimal up to logarithmic factors, and show in particular that one cannot compress $A$ to $O(d)$ ``directions'' $v_1, \ldots, v_{O(d)}$, such that for any $x$, $\|Ax\|_1$ can be well-approximated from $\langle v_1, x \rangle, \ldots, \langle v_{O(d)}, x \rangle$. Our lower bound rules out arbitrary functions of these inner products (and in fact arbitrary data structures built from $A$), and thus rules out the possibility of a singular value decomposition for $\ell_1$ in a very strong sense. Indeed, as $\epsilon \rightarrow 0$, for $p = 1$ the space complexity becomes arbitrarily large, while for $p = 2$ it is at most $O(d^2 \log(nd))$. As corollaries of our main lower bound, we obtain new lower bounds for a wide range of applications, including the above, which in many cases are optimal.  
\end{abstract}

\newpage

\tableofcontents

\newpage

%!TEX root = main.tex
\section{Introduction}
The explosive growth of available data has necessitated new models for
processing such data. A particularly powerful tool for analyzing such
data is {\it sketching}, which has found applications to communication complexity, data stream
algorithms, functional analysis, machine learning, numerical linear algebra, sparse recovery, and
many other areas. Here one is given a large object, such as a graph, a matrix, or a vector, 
and one seeks to compress it while still preserving useful information
about the object.
One of the main goals of a sketch is to use as little memory
as possible in order to compute functions of interest. Typically, to obtain
non-trivial space bounds, such sketches need to be both randomized and approximate.
By now there are nearly-optimal bounds on the memory required of sketching many 
fundamental problems, such as graph sparsification, norms of vectors, and problems
in linear algebra such as low-rank approximation and regression. 
We refer the reader to the surveys~\cite{m05,w14} as well as the compilation
of lecture notes~\cite{course}. 

In this paper we consider the {\it subspace sketch problem}.
\begin{definition}\label{def:SS}
Given an $n \times d$ matrix $A$ with entries specified by $O(\log(nd))$ bits, an accuracy
parameter $\epsilon > 0$, and a function
$\Phi: \mathbb{R}^n \rightarrow \mathbb{R}^{\geq 0}$, design a data structure $Q_\Phi$ so that for 
any $x \in \mathbb{R}^d$, with probability at least $0.9$, $Q_\Phi(x) = (1 \pm \epsilon)\Phi(Ax)$. 
\end{definition} 

The subspace sketch problem captures many important problems as special cases. 
We will show how to use this problem to bound the communication of approximating 
statistics of a matrix product, the size of coresets in projective clustering, 
the memory of streaming algorithms for regression in the row-update model, 
and the embedding dimension in functional analysis. We will describe these applications
in more detail below. 

The goal in this work is to determine the memory, i.e., the size of $Q_\Phi$, required
for solving the subspace sketch problem for different functions $\Phi$. We first consider the
classical $\ell_p$-norms $\Phi(x) = \sum_{i=1}^n |x_i|^p$, in which case the problem is referred to as the $\ell_p$ subspace sketch problem\footnote{Note we are technically considering the $p$-th power of the $\ell_p$-norms, but for the purposes of $(1+\epsilon)$-approximation, they are the same for constant $p$. Also, when $p < 1$, $\ell_p$ is not a norm, though it is still a well-defined quantity. Finally, $\ell_0$ denotes the number of non-zero entries of $x$.}. We later extend our techniques to their robust 
counterparts $\Phi(x) = \sum_{i=1}^n \phi(x_i)$, where $\phi(t) = |t|^p$ if $|t| \leq \tau$ and $\phi(t) = \tau^p$ otherwise.
Here $\Phi$ is a so-called $M$-estimator and known as 
the Tukey loss $p$-norm. It is less sensitive to ``outliers'' 
since it truncates large coordinate valus at $\tau$. 
We let $Q_p$ denote $Q_\Phi$ when $\Phi(x) = \sum_i |x_i|^p$, and use $Q_{p, \tau}$ when $\Phi$ is the 
Tukey loss $p$-norm.

It is known that for $p \in (0,2]$ and $r = O(\varepsilon^{-2})$, 
if one chooses a matrix $S \in \mathbb{R}^{r \times n}$ of i.i.d.\@ $p$-stable random variables, 
then for any fixed $y \in \mathbb{R}^n$, from the sketch $S \cdot y$ one can output a number 
$z$ for which
$(1-\epsilon)\|y\|_p \leq z \leq (1+\epsilon)\|y\|_p$ with probability at least $0.9$ \cite{i06}. 
We say $z$ is a $(1 \pm \epsilon)$-approximation of $\|y\|_p$. 
For $p = 1$, the output is just $\operatorname{med}(Sy)$, where $\operatorname{med}(\cdot)$ 
denotes the median of the absolute values of the coordinates in a vector. A sketch $S$ with
$r = O(\eps^{-2} \log n)$ rows is also known for $p = 0$ \cite{knw10}. For $p > 2$, there
is a distribution on $S \in \mathbb{R}^{r \times n}$ 
with $r = O(n^{1-2/p} \log n / \varepsilon^2)$ for which
one can output a $(1 \pm \epsilon)$-approximation of $\|y\|_p$ given $Sy$ with probability
at least $0.9$ \cite{gw18}. By appropriately discretizing the entries, 
one can solve the $\ell_p$ subspace sketch problem by storing $S \cdot A$ for an 
appropriate sketching matrix $S$, and estimating $\|Ax\|_p$ using $S \cdot A \cdot x$.
In this way,
one obtains a sketch of size 
$\widetilde{O}(\eps^{-2} d)$\footnote{Throughout we use 
$\widetilde{O}, \widetilde{\Omega},$ and $\widetilde{\Theta}$
to hide factors that are polynomial in $\log(nd/\epsilon)$. We note that 
our lower bounds are actually independent of $n$.} bits for $p \in [0, 2]$, and a sketch of size
$\widetilde{O}(n^{1-2/p} / \varepsilon^2 \cdot d)$ bits for $p > 2$. Note, however, that this is only one particular
approach, based on choosing a random matrix $S$, and better approaches may be possible. Indeed,
note that for $p = 2$, one can simply store $A^T A$ and output $Q_2(x) = x^TA^TAx$.
This is exact (i.e., holds for $\epsilon = 0$) and only uses $O(d^2 \log(nd))$ bits of space, 
which is significantly smaller than $\widetilde{O}(\epsilon^{-2} d)$ for small enough $\epsilon$. We note that the $\epsilon^{-2}$ term may be extremely prohibitive in applications.
For example, if one wants high accuracy such as $\epsilon = 0.1\%$, the $\epsilon^{-2}$ factor is a severe drawback of existing algorithms. 

A natural question is what makes it possible for $p = 2$ to obtain $\widetilde{O}(d^2)$ bits of
space, and whether it is also possible to achieve $\widetilde{O}(d^2)$ space for $p = 1$. One thing
that makes this possible for $p = 2$ is the singular value decomposition (SVD), namely, that
$A = U \Sigma V^T$ for matrices $U \in \mathbb{R}^{n \times d}$ and $V \in \mathbb{R}^{d \times d}$ 
with orthonormal columns, and $\Sigma$ is a non-negative
diagonal matrix. Then $\|Ax\|_2^2 = \|\Sigma V^T x\|_2^2$ since $U$ has orthonormal columns. Consequently,
it suffices to maintain the $d$ inner products 
$\langle \Sigma_{1,1} v_1, x \rangle, \ldots, \langle \Sigma_{d,d} v_d, x \rangle$, 
where the $v_i$'s are the rows of $V^T$. Thus one can ``compress'' $A$ to $d$ ``directions'' 
$\Sigma_{i,i} v_i$. A natural question is whether for $p = 1$ it is also possible find $O(d)$ 
directions $v_1, \ldots, v_{O(d)}$, 
such that for any
$x$, $\|Ax\|_1$ can be well-approximated from some function of 
$\langle v_1, x \rangle, \ldots, \langle v_{O(d)}, x \rangle$. Indeed, this would be the analogue of the 
SVD for $p = 1$, for which little is known. 

The central question of our work is: 
\begin{center}
{\it How much memory is needed to solve the subspace sketch problem as a function of $\Phi$?}
\end{center}

\subsection{Our Contributions}
Up to polylogarithmic factors, we resolve the above question for $\ell_p$-norms and Tukey
loss $p$-norms for any $p \in [0,2)$. For $p \geq 2$ we also obtain a surprising separation
for even integers $p$ from other values of $p$. 

Our main theorem is the following. We denote by $\Z^+$ the set of positive integers.
\begin{theorem}[Informal]\label{thm:main}
Let $p\in [0,\infty)\setminus 2\Z^+$ be a constant. 
For any $d = \Omega(\log(1/\varepsilon))$ and $n = \widetilde{\Omega}(\varepsilon^{-2} \cdot d)$, we have that $\widetilde{\Omega}(\varepsilon^{-2} \cdot d)$ bits are
necessary to solve the $\ell_p$ subspace sketch problem.
\end{theorem}

When $p \in 2\Z^+$, there is an upper bound of $O(d^p \log(nd))$ bits, independent of $\eps$ (see Remark~\ref{rem:eps=0_even d}). This gives a surprising separation between positive even integers and other values of $p$; in particular for positive even integers $p$ it is possible to obtain $\eps = 0$ with at most $O(d^p \log(nd))$ bits of space, whereas for other values of $p$ the space becomes arbitrarily large as $\eps \rightarrow 0$. This also shows it is not possible, 
for $p = 1$ for example, 
to find $O(d)$ representative directions 
for $\epsilon = 0$ analogous to the SVD for $p = 2$. Note that the lower bound in Theorem
\ref{thm:main} is much stronger than this, showing that there is no data structure whatsoever which
uses fewer than $\widetilde{\Omega}(\varepsilon^{-2} \cdot d)$ bits, and so as $\varepsilon$ gets smaller, the
space complexity becomes arbitrarily large. 

In addition to the $\ell_p$-norm, in the subspace sketch problem we also consider a more general entry-decomposable $\Phi$, that is, $\Phi(v) = \sum_i \phi(v_i)$ for $v\in \R^n$ and some $\phi:\R\to\R^{\geq 0}$. We show the same $\widetilde{\Omega}(\eps^{-2} \cdot d)$ lower bounds for a number of $M$-estimators $\phi$.

\begin{theorem}\label{thm:$M$-estimator_intro} The subspace sketch problem requires $\widetilde{\Omega}(\eps^{-2} d)$ bits when $d = \Omega(\log(1/\eps))$ and $n = \widetilde{\Omega}(\varepsilon^{-2} \cdot d)$ for the following functions $\phi$:
\begin{itemize}[topsep=0.5ex,itemsep=-0.5ex,partopsep=1ex,parsep=1ex]
	\item ($L_1$-$L_2$ estimator) $\phi(t) = 2(\sqrt{1+t^2/2}-1)$;
	\item (Huber estimator) $\phi(t) = t^2/(2\tau)\cdot \mathbf{1}_{\{|t|\leq \tau\}} + (|t|-\tau/2)\cdot\mathbf{1}_{\{|t| > \tau\}}$;
	\item (Fair estimator) $\phi(t) = \tau^2(|x|/\tau - \ln(1+|t|/\tau))$;
	\item (Cauchy estimator) $\phi(t) = (\tau^2/2)\ln(1+(t/\tau)^2)$;
	\item (Tukey loss $p$-norm) $\phi(t) = |t|^p\cdot \mathbf{1}_{\{|t|\leq \tau\}} + \tau^p\cdot\mathbf{1}_{\{|t| > \tau\}}$.
\end{itemize}
\end{theorem}
We also consider a mollified version of the $1$-Tukey loss function, for which the lower bound of $\widetilde{\Omega}(\eps^{-2}d)$ bits still holds. Furthermore, the lower bound is tight up to logarithmic factors, since we design a new algorithm which approximates $\Phi(x)$ using $\widetilde{O}(\eps^{-2})$ bits, which implies an upper bound of $\widetilde{O}(\eps^{-2} d)$ for the subspace sketch problem. See Section~\ref{sec:tukey_UB} for details.

While Theorem \ref{thm:main} gives a tight lower bound for $p \in [0,2)$, matching
the simple sketching upper bounds described earlier, and also gives a separation from the 
$O(d^p \log(nd))$ bit bound for even
integers $p \geq 2$, one may ask what exactly the space required is for even integers
$p \geq 2$ and arbitrarily small $\epsilon$. 
For $p = 2$, the $O(d^2 \log(nd))$ upper bound 
is tight up to logarithmic factors since the previous work \cite[Theorem 2.2]{andoni2014sketching}
implies an $\widetilde{\Omega}(d^2)$ lower bound once $\epsilon = O(1/\sqrt{d})$. For $p > 2$, we show
the following: for a constant $\eps \in (0,1)$, there is an upper bound of 
$\widetilde{O}(d^{p/2})$ bits (see Remark~\ref{rem:d^{p/2}_foreach_tight}), 
which is nearly tight in light of the following lower bound, which
holds for constant $\eps$. 
\begin{theorem}[Informal]\label{thm:even_p_lb}
Let $p \ge 2$ and $\eps\in (0,1)$ be constants. Suppose that $n=\widetilde{\Omega}(d^{p/2})$, then
$\widetilde{\Omega}(d^{p/2})$ bits are
necessary to solve the $\ell_p$ subspace sketch problem.
\end{theorem}
Note that Theorem \ref{thm:even_p_lb} holds even if $p$ is not an even integer, and shows that a lower bound of $d^{\Omega(p)}$ holds for every $p \geq 2$.   

We next turn to concrete applications of Theorems~\ref{thm:main} and~\ref{thm:$M$-estimator_intro}.

\paragraph{Statistics of a Matrix Product.} %It is widely known that the set-intersection join of databases is naturally related to Boolean matrix multiplication. 
%Suppose $A_1,\dots,A_n$ and $B_1,\dots, B_n$ are subsets of $[n]$ and the set-intersection join asks to find the pairs $(i,j)$ such that $A_i\cap B_j \neq \emptyset$. 
In \cite{WZ18}, an algorithm was given for estimating $\|A \cdot B\|_p$ for integer matrices $A$ and $B$ with $O(\log n)$ bit integer entries (see Algorithm 1 in \cite{WZ18} for the general algorithm). When $p = 0$, this estimates the number of non-zero entries of $A \cdot B$, which may be useful since there are faster algorithms for matrix product when the output is sparse, see \cite{p13} and the references therein. 
More generally, norms of the product $A \cdot B$ can be used to determine how correlated the rows of $A$ are with the columns of $B$. The bit complexity of this problem was studied in~\cite{GWWZ15, WZ18}. In~\cite{GWWZ15} a lower bound of $\Omega(\eps^{-2} \cdot n)$ bits was shown for estimating $\|AB\|_0$ for $n\times n$ matrices $A,B$ up to a $(1+\eps)$ factor, assuming $n \geq 1/\eps^2$ (this lower bound holds already for binary matrices $A$ and $B$). This lower bound implies an $\ell_0$-subspace sketch lower bound of $\Omega(\eps^{-2} \cdot d)$ {\it assuming that} $d\geq 1/\eps^2$. Our lower bound in Theorem~\ref{thm:main} considerably strengthens this result by showing the same lower bound (up to $\polylog(d/\eps)$ factors) for a much smaller value of $d = \Omega(\log(1/\eps))$. For any $p \in [0,2]$, there is a matching upper bound up to polylogarithmic factors (such an upper bound is given implicitly in the description of Algorithm 1 of \cite{WZ18}, where the $\epsilon$ there is instantiated with $\sqrt{\epsilon}$, and also follows from the random sketching matrices $S$ discussed above). 

\paragraph{Projective Clustering.}
In the task of projective clustering, we are given a set $X \subset \mathbb{R}^d$ of $n$ points, a positive integer $k$, and a non-negative integer $j \le d$. 
A center $\cC$ is a $k$-tuple $(V_1, V_2, \ldots, V_k)$, where each $V_i$ is a $j$-dimensional affine subspace in $\mathbb{R}^d$. Given a function $\phi:\R\to\R^{\geq 0}$, the objective is to find a center $\cC$ that minimizes the projective cost, defined to be 
\[
\cost(X, \cC) = \sum_{x \in X} \phi(\dist(x, \cC)),
\] 
where $\dist(x, \cC) = \min_i \dist(x, V_i)$, the Euclidean distance from a point $p$ to its nearest subspace $V_i$ in $\cC = (V_1, V_2, \ldots, V_k)$. The coreset problem for projective clustering asks to design a data structure $Q_\phi$ such that for any center $\cC$, with probability at least $0.9$, $Q_\phi(\cC) = (1\pm \eps)\cost(X,\cC)$. Note that in this and other computational geometry problems, the dimension $d$ may be small (e.g., $d = \log(1/\epsilon$)), though one may want a high accuracy solution. Although possibly far from optimal, surprisingly our lower bound below is the first non-trivial lower bound on the size of coresets for projective clustering.

\begin{theorem}[Informal]\label{thm:projective_clustering_intro}
Suppose that $\phi(t) = |t|^p$ for $p\in [0,\infty)\setminus 2\Z^+$ or $\phi$ is one of the functions in Theorem~\ref{thm:$M$-estimator_intro}.
For $k\geq 1$ and $j = \Omega(\log(k/\eps))$, any coreset for projective clustering requires $\widetilde{\Omega}(\eps^{-2} k j)$ bits.
\end{theorem}

\paragraph{Linear Regression.}
In the linear regression problem, there is an $n \times d$ data matrix $A$ and a vector $b \in \mathbb{R}^n$. 
The goal is to find a vector $x \in \mathbb{R}^d$ so as to minimize $\Phi(Ax - b)$, where $\Phi(v) = \sum_i \phi(v_i)$ for $v\in \R^n$ and some $\phi:\R\to\R^{\geq 0}$.
Here we consider streaming coresets for linear regression in the row-update model.
In the row-update model, the streaming coreset is updated online during one pass over the $n$ rows of $\begin{pmatrix}A & b\end{pmatrix}$, and outputs a $(1 \pm \varepsilon)$-approximation to the optimal value $\min_x \Phi(Ax - b)$ at the end. 
By a simple reduction, our lower bound for the subspace sketch problem implies lower bounds on the size of streaming coresets for linear regression in the row-update model.
To see this, we note that by taking sufficiently large $\lambda$, 
$$\min_y \left( \Phi(Ay) + \lambda \Phi(x - y)\right) = \Phi(Ax).$$
Thus, a streaming coreset for linear regression can solve the subspace sketch problem, which we formalize in the following corollary.

\begin{corollary}
Suppose that $\phi(t) = |t|^p$ for $p\in [0,\infty)\setminus 2\Z^+$ or $\phi$ is one of the functions in Theorem~\ref{thm:$M$-estimator_intro}.
Any streaming coreset for linear regression in the row-update model requires $\widetilde{\Omega}\left(\varepsilon^{-2} d \right)$ bits when $d = \Omega(\log(1/\eps))$.
\end{corollary}

\paragraph{Subspace Embeddings.} Let $p\geq 1$. Given $A\in \R^{n\times d}$, the $\ell_p$ subspace embedding problem asks to find a linear map $T:\R^n\to \R^r$ such that for all $x \in \mathbb{R}^d$, 
\begin{equation}\label{eqn:l1ss}
(1-\eps)\|Ax\|_p \leq \|TAx\|_p\leq (1 + \eps)\|Ax\|_p.
\end{equation}
The smallest $r$ which admits a $T$ for every $A$ is denoted by $N_p(d,\eps)$, which is of great interest in functional analysis. When $T$ is allowed to be random, we require \eqref{eqn:l1ss} to hold with probability at least $0.9$. This problem can be seen as a special case of the ``for-all'' version of the subspace sketch problem in Definition~\ref{def:SS}. 
In the for-all version of the subspace sketch problem, the data structure $Q_p$ is 
required to, with probability at least $0.9$, satisfy $Q_p(x) = (1\pm \eps)\|Ax\|_p$ 
simultaneously for all $x\in\R^d$.
In this case, the same lower bound of $\widetilde{\Omega}(\eps^{-2} \cdot d)$ bits holds 
for $p\in [1,\infty)\setminus 2\Z$. 

Since the data structure can store $T$ if it exists, we can turn our bit lower bound into a dimension lower bound on $N_p(d,\eps)$. Doing so will incur a loss of an $\widetilde{O}(d)$ factor (Theorem~\ref{thm:conversion_to_dimension_lb}). 
We give an $\widetilde\Omega(\eps^{-2})$ lower bound, which is the first such lower bound giving a dependence on $\eps$ for general $p$. 

\begin{corollary}\label{cor:N_p(d,eps)}
Suppose that $p\in [1,\infty)\setminus 2\Z$ and $d=\Omega(\log(1/\eps))$. It holds that $N_p(d,\eps) = \widetilde\Omega(\eps^{-2})$.
\end{corollary}

The dependence on $\eps$ in this lower bound is tight, up to $\polylog(1/\eps)$ factors, for all values of $p \in [1, \infty) \setminus 2\Z$ ~\cite{schechtman:l_r^n}.  When $p\in 2\Z$, no lower bound with a dependence on $\eps$ should exist, since a $d$-dimensional subspace of $\ell_p^n$ always embeds into $\ell_p^r$ isometrically with $r = \binom{d+p-1}{p}-1$~\cite{handbook:21}. See more discussion below in Section~\ref{sec:func_anal} on functional analysis.
We also prove a bit complexity lower bound for the aforementioned for-all version of the subspace sketch problem. 
We refer the reader to Section~\ref{sec:forall_lb_poly(d)} for details.

\begin{theorem}
Let $p\geq 1$ be a constant. Suppose that $\eps > 0$ is a constant. The for-all version of the subspace sketch problem requires $\Omega(d^{\max\{p/2,1\}+1})$ bits.
\end{theorem}

This lower bound immediately implies a dimension lower bound of $N_p(d,\eps)=\widetilde{\Omega}(d^{\max\{p/2,1\}})$ for the subspace embedding problem for constant $\eps$, recovering existing lower bounds (up to logarithmic factors), which are known to be tight.

\paragraph{Sampling by Lewis Weights.} While it is immediate that $N_p(d,\eps)\geq d$, our lower bound above thus far has not precluded the possibility that $N_p(d,\eps) = \widetilde{O}(d + 1/\eps^2)$. However, the next corollary, which lower bounds the target dimension for sampling-based embeddings, indicates this is impossible to achieve using a prevailing existing technique. 

\begin{corollary}
Let $p\geq 1$ and $p\notin 2\Z$. 
Suppose that $Q_p(x) = \|TAx\|_p^p$ solves the $\ell_p$ subspace sketch problem for some $T\in \R^{r\times n}$ for which each row of $T$ contains exactly one non-zero element. Then $r = \widetilde\Omega(\eps^{-2} d)$, provided that $d = \Omega(\log(1/\eps))$.
\end{corollary}

The same lower bound holds for the for-all version of the $\ell_p$ subspace sketch problem. 
As a consequence, since the upper bounds of $N_p(d,\eps)$ in \eqref{eqn:N_ub} for $1\leq p<2$ are based on subsampling with the ``change of density'' technique (also known as sampling by Lewis weights~\cite{lewis_sampling}), they are, within the framework of this classical technique, best possible up to polylog$(d/\epsilon)$ factors. 

\paragraph{Oblivious Sketches.}
For the for-all version of the $\ell_p$ subspace sketch problem, we note that there exist general sketches such as the Cauchy sketch~\cite{clarkson2016fast} which are beyond the reach of the corollary above. Note that the Cauchy sketch is an oblivious sketch, which means the distribution is independent of $A$. We also prove a dimension lower bound of $\widetilde{\Omega}(\eps^{-2} \cdot d)$ on the target dimension for oblivious sketches (see Section~\ref{sec:OSE}), which is tight up to logarithmic factors since the Cauchy sketch has a target dimension of $O(\eps^{-2} d \log(d/\eps))$.

\begin{theorem}[Informal]
Let $p\in [1,2)$ be a constant. Any oblivious sketch that solves the for-all version of the $\ell_p$ subspace sketch problem has a target dimension of $\widetilde{\Omega}(\eps^{-2} d)$.
\end{theorem}

Therefore, it is natural to ask in general whether $N_p(d,\eps) = \widetilde{\Omega}(d/\eps^2)$. A proof using the framework of this paper would require an $\widetilde\Omega(d^2/\eps^2)$ lower bound for the for-all version of the $\ell_p$ subspace sketch problem. We conjecture it is true; however, our current methods, giving almost-tight lower bounds (in the for-each sense), do not extend to give this result and so we leave it as a main open problem.
 
\subsection{Connection with Banach Space Theory}\label{sec:func_anal}
In the language of functional analysis, the $\ell_p$ subspace embedding problem is a classical problem in the theory of $L_p$ spaces with a rich history.
For two Banach spaces $X$ and $Y$, we say $X$ $K$-embeds into $Y$, 
if there exists an injective homomorphism $T:X\to Y$ satisfying $\|x\|_X\leq \|Tx\|_Y\leq K\|x\|_X$ for all $x\in X$. Such a $T$ is called an \emph{isomorphic embedding}. A classical problem in the theory of Banach spaces is to consider the isomorphic embedding of finite-dimensional subspaces of $L_p = L^p(0,1)$ into $\ell_p^n = (\R^n,\|\cdot\|_p)$, where $p\geq 1$ is a constant. Specifically, the problem asks what is the minimum value of $n$, denoted by $N_p(d,\eps)$, for which all $d$-dimensional subspaces of $L_p$ $(1+\eps)$-embed into $\ell_p^n$. A comprehensive survey of this problem can be found in~\cite{handbook:19}. 

The case of $p=2$ is immediate, in which case one can take $n=d$ and $\eps = 0$, obtaining an isometric embedding, and thus we assume $p\neq 2$. We remark that, when $p$ is an even integer, it is also possible to attain an isometric embedding into $\ell_p^n$ with $n=\binom{d+p-1}{p}-1$~\cite{handbook:21}. In general, the best\footnote{A few upper bounds for the case $p\in (2,\infty)\setminus 2\Z$ are known, none of which dominates the rest. Here we choose the one having the best dependence on both $d$ and $\eps$, up to $\polylog(d/\eps)$ factors.} 
known upper bounds on $N_p(d,\eps)$ are as follows.
\begin{equation}\label{eqn:N_ub}
N_p(d,\eps) \leq \begin{cases}
								C\eps^{-2}d \log d &\quad p=1\\
								C\eps^{-2}d (\log \eps^{-2}d)(\log \log\eps^{-2}d+\log(1/\eps))^2 & \quad p\in (1,2)\\
							 	C_p\eps^{-2}d^{p/2}\log^2 d\log(d/\eps) & \quad p\in (2,\infty)\setminus 2\Z\\
							 	C\eps^{-2} (10d/p)^{p/2}& \quad p\in 2\Z,
							\end{cases}
\end{equation}
where $C > 0$ is an absolute constant and $C_p>0$ is a constant that depends only on $p$.
The cases of $p=1$ and $p\in (1,2)$ are due to Talagrand~\cite{talagrand:p=1,talagrand:1<p<2}.
The case of non-even integers $p>2$ is % due to Bourgain et al.~\cite{bourgain:zonoid_lb} and the bound above 
taken from~\cite[Theorem 15.13]{LT91}, based on the earlier work of Bourgain et al.~\cite{bourgain:zonoid_lb}.
The case of even integers $p$ is due to Schechtman~\cite{schechtman:tight}.

The upper bounds in~\eqref{eqn:N_ub} are established by subsampling with a technique called ``change of density''~\cite{handbook:19}. First observe that it suffices to consider embeddings from $\ell_p^N$ to $\ell_p^n$ since any $d$-dimensional subspace of $L_p$ $(1+\eps)$-embeds into $\ell_p^N$ for some large $N$. Now suppose that $E$ is a $d$-dimensional subspace of $\ell_p^N$. One can show that randomly subsampling coordinates induces a low-distortion isomorphism between $E$ and $E$ restricted onto the sampled coordinates, provided that each element of $E$ is ``spread out'' among the coordinates, which is achieved by first applying the technique of change of density to $E$. 

Regarding lower bounds, a quick lower bound follows from the tightness of Dvoretzky's Theorem for $\ell_p$ spaces (see, e.g.~\cite[p21]{MS:book}), which states that if $\ell_2^d$ $2$-embeds into $\ell_p^n$, then $n\geq c d$ for $1\leq p < 2$ and $n\geq (cd/p)^{p/2}$ for $p\geq 2$, where $c>0$ is an absolute constant. Since $\ell_2^d$ embeds into $L_p$ isometrically for all $p\geq 1$~\cite[p16]{handbook:1}, identical lower bounds for $N_p(d,\eps)$ follow. Hence the upper bounds in~\eqref{eqn:N_ub} are, in terms of $d$, tight for $p\in 2\Z$ and tight up to logarithmic factors for other values of $p$. However, the right dependence on $\eps$ is a long-standing open problem and little is known. 
See~\cite[p845]{handbook:19} for a discussion on this topic. 
It is known that $N_1(d,\eps)\geq c(d)\eps^{-2(d-1)/(d+2)}$~\cite{bourgain:zonoid_lb}, whose proof critically relies upon the fact that the unit ball of a finite-dimensional space of $\ell_1$ is the polar of a zonotope (a linear image of cube $[-1,1]^d$) and the $\ell_1$-norm for vectors in the subspace thus admits a nice representation~\cite{bolker}.
However, a lower bound for general $p$ is unknown. Our Corollary~\ref{cor:N_p(d,eps)} shows that $n\geq c \eps^{-2}/\poly(\log(1/\varepsilon))$ for all $p \geq 1$ and $p\not\in 2\Z$, which is the first lower bound on the dependence of $\eps$ for general $p$, and is optimal up to logarithmic factors. We would like to stress that except for the very special case of $\ell_1$, no lower bound on the dependence on $\epsilon$ whatsoever was known for $p \not\in 2\Z$. We consider this to be significant evidence of the generality and novelty of our techniques. Moreover, even our lower bound for $p = 1$ is considerably wider in scope, as discussed more below.

\subsection{Comparison with Prior Work}
\subsubsection{Comparison with Previous Results in Functional Analysis}
As discussed, the mentioned lower bounds on $N_p(d,\eps)$ come from the tightness of Dvoretzky's Theorem, which shows the impossibility of embedding $\ell_2^d$ into a Banach space with low distortion. Here the hardness comes from the geometry of the target space. In contrast, we emphasize that the hardness in our $\ell_p$ subspace sketch problem comes from the source space, since the target space is unconstrained and the output function $Q_p(\cdot)$ does not necessarily correspond to an embedding. 
The lower bound via tightness of Dvoretzky's Theorem cannot show that $\ell_p^d$ does not $(1 + \eps)$-embed into $\ell_q^{n}$ for $d=\Theta(\log(1/\eps))$ and $n = O(1 / \eps^{1.99})$, where $q\not\in 2\Z$. 

When the target space is not $\ell_p$, lower bounds via functional analysis are more difficult to obtain since they require understanding the geometry of the target space. Since our data structure problem has no constraints on $Q_p(\cdot)$, the target space does not even need to be normed. In theoretical computer science and machine learning applications, the usual ``sketch and solve'' paradigm typically just requires the target space to admit an efficient algorithm for the optimization problem at hand\footnote{For example, consider the space $\R^n$ endowed with a premetric $d(x,y)=\sum_i f(x_i-y_i)$, where $f(x)=\tau x \mathbf{1}_{\{x\geq 0\}} + (\tau-1)x\mathbf{1}_{\{x\leq 0\}}$ ($\tau\in(0,1)$), which is not even symmetric when $\tau\neq \frac{1}{2}$. See~\cite{YMM} for an embedding into this space.
}. Our lower bounds are thus much wider in scope than those in geometric functional analysis.

\subsubsection{Comparison with Previous Results for Graph Sparsifiers}
Recently, the bit complexity of cut sparsifiers was studied in \cite{andoni2014sketching, carlson2017optimal}. 
Given an undirected graph $G = (V, E)$, $|V|=d$, a function $f : 2^{V} \to \R$ is a $(1 + \eps)$-{\em cut sketch}, if for any vertex set $S \subseteq V$, 
$$
(1-\eps) C(S, V\setminus S) \leq f(S) \leq (1 + \eps) C(S, V \setminus S),
$$
where $C(S, V \setminus S)$ denotes the capacity of the cut between $S$ and $V \setminus S$.
The main result of these works is that any $(1 + \eps)$-cut sketch requires $\Omega(\eps^{-2} d \log d)$ bits to store. 
Note that a cut sketch can be constructed using a for-all version of the $\ell_p$ subspace sketch for any $p$, by just taking the matrix $A$ to be the edge-vertex incidence matrix of the graph $G$ and querying all vectors $x \in \{0, 1\}^d$.
Thus, one may naturally ask if the lower bounds in \cite{andoni2014sketching, carlson2017optimal} imply any lower bounds for the subspace sketch problem. 

We note that both works \cite{andoni2014sketching, carlson2017optimal} have explicit constraints on the value of $\eps$.
In \cite{andoni2014sketching}, in order to prove the $\Omega(\eps^{-2} d)$ lower bound, it is required that $\eps = \Omega(1 / \sqrt{d})$.
In \cite{carlson2017optimal} the lower bound of $\Omega(d\log d/\varepsilon^2)$ requires $\eps = \omega(1 / d^{1 / 4})$.
Thus, the strongest lower bound that can be proved using such an approach is $\widetilde{\Omega}(d^2)$. 
This is natural, since one can always store the entire adjacency matrix of the graph in $\widetilde{O}(d^2)$ bits.
Our lower bound, in contrast, becomes arbitrarily large as $\varepsilon \to 0$.
% !TEX root = main.tex
\subsection{Our Techniques}

We use the case of $p = 1$ to illustrate our ideas behind the $\widetilde{\Omega}\left(\eps^{-2}\right)$ lower bound for the $\ell_p$ subspace sketch problem, when $d = \Theta(\log(1/\varepsilon))$. We then extend this to an $\widetilde{\Omega}\left(\eps^{-2} d \right )$ lower bound for general $d$ via a simple padding argument. 
We first show how to prove a weaker $\widetilde{\Omega}\left(\eps^{-1} \right)$ lower bound for the for-all version of the problem, and then show how to strengthen the argument to obtain both a stronger $\widetilde{\Omega}\left(\eps^{-2}\right)$ lower bound and in the weaker original version of
the problem (the ``for-each'' model, where we only need to be correct on a fixed query $x$ with
constant probability).

Note that the condition that $d = \Theta(\log(1 / \varepsilon))$ is crucial for our proof. As shown in Section~\ref{sec:2dub}, when $d = 2$, there is actually an $\widetilde{O}(\varepsilon^{-1})$ upper bound, and thus our $\widetilde{\Omega}(\varepsilon^{-2})$ lower bound does not hold universally for all values of $d$. It is thus crucial that we look at a larger value of $d$, and we show that
$d = \Theta(\log(1/\varepsilon))$ suffices. 

To prove our bit lower bounds for the $\ell_1$ subspace sketch problem, we shall encode random bits in the matrix $A$ such that  having a $(1+\eps)$-approximation to $\|Ax\|_1$ will allow us to recover, in the for-each case, some specific random bit, and in the for-all case, all the random bits using different choices of $x$. A standard information-theoretic argument then implies that the lower bound for the subspace sketch problem is proportional to the number of random bits we can recover.

\paragraph{Warmup: An $\widetilde{\Omega}\left(\eps^{-1} \right)$ Lower Bound for the For-All Version.}
In our hard instance, we let $d = \Theta(\log(1 / \eps))$ be such that $n = 2^d = \widetilde{\Theta}(1/\eps)$. 
Form a matrix $A \in \R^{n \times d}$ by including all vectors $i \in \bcube^d$ as its rows and then scaling the $i$-th row by a nonnegative scalar $r_i \le \poly(d)$.
We can think of $r$ as a vector in $\R^{n}$ with $\|r\|_{\infty} \le \poly(d)$.
Now, we query $Q_1(i)$ for all vectors $i \in \bcube^d$. 
For an appropriate choice of $d = \Theta(\log(1 / \eps))$, for all $i \in \bcube^d$, we have 
\begin{equation}\label{eqn:|Ai|_1}
\|Ai\|_1 = \sum_{j \in \bcube^d} r_j \cdot |\langle i, j \rangle| \leq 2^d \cdot \poly(d) < \frac{1}{\eps}.
\end{equation}
Since $Q_1(i)$ is a $(1 \pm \eps)$-approximation to $\|Ai\|_1$, and $\|Ai\|_1$ is always an integer, we can recover the exact value of $\|Ai\|_1$ using $Q_1(i)$, for all $i \in \bcube^d$.

Now we define a matrix $M \in \R^{n \times n}$, where $M_{i, j} = |\langle i, j \rangle|$, where $i, j$ are interpreted as vectors in $\bcube^d$.
A simple yet crucial observation is that, $\|Ai\|_1$ is exactly the $i$-th coordinate of $M r$. Notice that this critically relies on the assumption that $r$ has nonnegative coordinates. 
Thus, the problem can be equivalently viewed as designing a vector $r \in \R^{n}$ with $\|r\|_{\infty} \le \poly(d)$ and recovering $r$ from the vector $M r$. 
At this point, a natural idea is to show that the matrix $M$ has a sufficiently large rank, say, $\rank(M) =  \widetilde{\Omega}(\eps^{-1})$, and carefully design $r$ to show an $\Omega(\rank(M)) =  \widetilde{\Omega}(\eps^{-1})$ lower bound. 

Fourier analysis on the hypercube shows that the eigenvectors of $M$ are the rows of the normalized Hadamard matrix, while the eigenvalues of $M$ are the Fourier coefficients associated with the function $g(s) = |d - 2w_H(s)|$, where $w_H(s)$ is the Hamming weight of a vector $s \in \F_2^d$.
Considering all vectors of Hamming weight $d / 2$ in $\F_2^d$ and their associated Fourier coefficients, we arrive at the conclusion that there are at least $\binom{d}{d/ 2}$ eigenvalues of $M$ with absolute value
\[
\left| \sum_{\substack{0\leq i\leq d\\ \text{$i$ is even}}} (-1)^{i/2} \binom{d / 2}{i / 2}  |d - 2i|\right|,
\]
which can be shown to be at least $\Omega(2^{d / 2} / \poly(d))$. The formal argument is given in Section~\ref{sec:spec}.
Hence $\rank(M) \ge \binom{d}{d / 2} = \Omega(2^d / \poly(d))$.
Without loss of generality we assume the $\rank(M) \times \rank(M)$ upper-left block of $A$ is non-singular. 

Now an $\widetilde{\Omega}(1 / \eps)$ lower bound follows readily. Set $r$ so that 
$$
r_i = \begin{cases}
s_i, & i \le \rank(M); \\
0, & i > \rank(M),
\end{cases}
$$
where $\{s_i\}_{i = 1}^{\rank(M)}$ is a set of \iid Bernoulli random variables.
Since the exact value of $M r$ is known and the $\rank(M) \times \rank(M)$ upper-left block of $A$ is non-singular, one can recover the values of $\{s_i\}_{i = 1}^{\rank(M)}$ by solving a linear system, which implies an $\Omega(\rank(M)) =  \widetilde{\Omega}(\eps^{-1})$ lower bound. 

Before proceeding, let us first review why our argument fails for $p = 2$. 
For the $\ell_p$-norm, the Fourier coefficients associated with the vectors of Hamming weight $d / 2$ on the Boolean cube are
\[
\left| \sum_{\substack{0\leq i\leq d\\ \text{$i$ is even}}} (-1)^{i/2} \binom{d / 2}{i / 2}  |d - 2i| ^p\right| = \Theta\left( \frac{ 2^{d / 2} }{ \sqrt{d}} \left|\sin\frac{p \pi}{2}\right| \right).
\]
Therefore this sum vanishes if and only if $p$ is an even integer, in which case $\rank(A)$ will no longer be $\Omega(2^d / \poly(d))$ and the lower bound argument will fail.

\paragraph{An $\widetilde{\Omega}\left(\eps^{-2}\right)$ Lower Bound for the For-Each Version.}
To strengthen this to an $\widetilde{\Omega}(\eps^{-2})$ lower bound, it is tempting to increase $d$ so that $n = 2^d = \widetilde{\Omega}(\eps^{-2})$.
In this case, however, we can no longer recover the exact value of $M r$, since each entry of $M r$ now has magnitude $\widetilde{\Theta}(\eps^{-2})$ and the function $Q_1(\cdot)$ only gives a $(1 \pm \eps)$-approximation. 
We still obtain a noisy version of $M r$, but with a $\widetilde{\Theta}(1 / \eps)$ additive error on each entry. 
One peculiarity of the model here is that if some entries of $r$ are negative, then $\|Ai\|_1 = (M|r|)_i$ (cf.~\eqref{eqn:|Ai|_1}), where $|r|$ denotes the vector formed by taking the absolute value of each coordinate of $r$, i.e., $\|Ai\|_1$ depends only on the absolute values of entries of $r$, which suggests that the constraint that each entry of $M r$ has magnitude $\widetilde{\Theta}(1 / \eps^2)$ with an additive error of $\widetilde{\Theta}(1 / \eps)$ is somehow intrinsic. 

To illustrate our idea for overcoming the issue of large additive error, for the time being let us forget the actual form of $M$ previously defined in the argument for our $\widetilde{\Omega}\left(\eps^{-1} \right)$ lower bound and consider instead a general $M\in \R^{n\times n}$ with orthogonal rows, each row having $\ell_2$ norm $\Omega(2^{d / 2} / \poly(d))$.
For now we also allow $r$ to contain negative entries such that $\|r\|_{\infty} \le \poly(d)$, and pretend that the noisy version of $Mr$ has an $\widetilde{\Theta}(1 / \eps)$ additive error on each entry. 
Now, let
$$
r = \sum_{i=1}^{n} s_i \cdot \frac{M_i}{\|M_i\|_2},
$$
where $\{s_i\}_{i=1}^{n}$ is a set of \iid Rademacher random variables. 
By a standard concentration inequality, $\|r\|_{\infty} \le  \poly(d)$ holds with high probability (recall that $n=2^d$). Consider the vector $M r$. Due to the orthogonality of the rows of $M$, the $i$-th coordinate of $M r$ will be
$$
\langle M_i, r\rangle = s_i \cdot \|M_i\|_2.
$$
Provided that $ \|M_i\|_2$ is larger than the additive error $\widetilde{\Theta}(1 / \eps)$, we can still recover $s_i$ by just looking at the sign of $\langle M_i, r\rangle$.
Thus, for an appropriate choice of $d$ such that $2^{d / 2}  / \poly(d) = \widetilde{\Omega}(1 / \eps)$, we can obtain an $\Omega(2^d) = \widetilde{\Omega}(1 / \eps^2)$ lower bound. 

Now we return to the original $M$ with $M_{i,j}=|\langle i,j\rangle|$, whose rows are not necessarily orthogonal. The previous argument still goes through so long as we can identify a subset $\mathcal{R} \subseteq [n] = [2^d]$ of size $|\mathcal{R}| \ge \Omega(2^d / \poly(d))$ such that the rows $\{M_i\}_{i\in\mathcal{R}}$ are nearly orthogonal, meaning that the $\ell_2$ norm of the orthogonal projection of $M_i$ onto the subspace spanned by other rows $\{M_j\}_{j \in \mathcal{R} \setminus \{i\}}$ is much smaller than $\|M_i\|_2$.

To achieve this goal, we study the spectrum of $M$, and as far as we are aware, this is the first such study of spectral properties of this matrix. 
The Fourier argument mentioned above implies that at least $\Omega(2^d / \poly(d))$ eigenvalues of $A$ have the same absolute value $\Omega(2^{d  / 2} / \poly(d))$.
If all other eigenvalues of $A$ were zero, then we could identify a set of $|\mathcal{R}| \ge \Omega(2^d / \poly(d))$ nearly orthogonal rows using rows of $A$ each with $\ell_2$ norm $\Omega(2^{d / 2} / \poly(d))$, using a procedure similar to the standard Gram-Schmidt process.
The full details can be found in Section \ref{sec:or}.
Although the other eigenvalues of $M$ are not all zero, we can simply ignore the associated eigenvectors since they are orthogonal to the set of nearly orthogonal rows we obtain above. 

Lastly, recall that what we truly obtain is $M|r|$ rather than $Mr$ unless $r\geq 0$. To fix this, note that $\|r\|_{\infty} \le \poly(d)$ with high probability, and so we can just shift each coordinate of $r$ by a fixed amount of $\poly(d)$ to ensure that all entries of $r$ are positive. 
We can still obtain $\langle M_i, r \rangle$ with an additive error $\widetilde{\Theta}(1 / \eps)$, since the amount of the shift is fixed and bounded by $\poly(d)$.

Notice that the above argument in fact holds even for the for-each version of the subspace sketch problem.
By querying the $i$-th vector on the Boolean cube for some $i \in \mathcal{R}$, we are able to recover the sign of $s_i$ with constant probability. 
Given this, a standard information-theoretic argument shows that our 
lower bound holds for the for-each version of the problem.

The formal analysis given in Section \ref{sec:comm_lb} is a careful combination of all the ideas mentioned above.

\paragraph{Applications: $M$-estimators and Projective Clustering Coresets.}
Our general strategy for proving lower bounds for $M$-estimators is to relate one $M$-estimator, for which we want to prove a lower bound, to another $M$-estimator for which a lower bound is easy to derive.
For the $L_1$-$L_2$ estimator, the Huber estimator and the Fair estimator, when $|t|$ is sufficiently large, $\phi(t) = (1 \pm \varepsilon) |t|$ (up to rescaling of $t$ and the function value), and thus the lower bounds follow from those for the $\ell_1$ subspace sketch problem.

For the Cauchy estimator, we relate it to another estimator $\phi_{\text{aux}}(t) = \ln |x|\cdot \mathbf{1}_{\{|x|\geq 1\}}$. 
In Section \ref{sec:m-est}, we show that our Fourier analytic arguments also work for $\phi_{\text{aux}}(t)$.
Since for sufficiently large $t$, the Cauchy estimator satisfies $\phi(t) = (1 \pm \varepsilon) \phi_{\text{aux}}(t)$ (up to rescaling of $t$ and the function value), a lower bound for the Cauchy estimator follows.

To prove lower bounds on coresets for projective clustering, the main observation is that when $k = 1$ and $j = d - 1$, by choosing the query subspace to be the orthogonal complement of a vector $z$, the projection cost is just $\sum_{x \in X}\phi(\langle x, z \rangle)$, and thus we can invoke our lower bounds for the subspace sketch problem.
We use a coding argument to handle general $k$.
In Lemma~\ref{lem:code}, we show there exists a set $S = \{(s_1, t_1), (s_2, t_2), \ldots, (s_{k}, t_{k})\}$, where $s_i, t_i \in \mathbb{R}^{O(\log k)}$, $\langle s_i, t_i \rangle = 0$ and $\langle s_i, t_j \rangle$ is arbitrarily large when $i \neq j$. 
Now for $k$ copies of the hard instance of the subspace sketch problem, we add $s_i$ as a prefix to all data points in the $i$-th hard instance, and set the query subspace to be the orthogonal complement of a vector $z$, to which we add $t_i$ as a prefix. Now, the data points in the $i$-th hard instance will always choose the $i$-th center in the optimal solution, since otherwise an arbitrarily large cost will incur. Thus, we can solve $k$ independent copies of the subspace sketch problem, and the desired lower bound follows. 

\medskip

In the rest of the section, we shall illustrate our techniques for proving lower bounds that depend on $p$ for the $\ell_p$ subspace sketch problem. These lower bounds hold even when $\varepsilon$ is a constant.  We again resort to information theory, trying to recover, using $Q_p$ queries, the entire matrix $A$ among a collection $\mathcal{S}$ of matrices. The lower bound is then $\Omega(\log|\mathcal{S}|)$ bits.
 
\paragraph{An $\widetilde{\Omega}(d^{p / 2})$ Lower Bound for the For-Each Version.}
Our approach for proving the $\widetilde{\Omega}(d^{p / 2})$ lower bound is based on the following crucial observation: consider a uniformly random matrix $A \in \{-1,1\}^{\Theta(d^{p / 2}) \times d}$ and a uniformly random vector $x \in \{-1, 1\}^d$. Then $\E\|Ax\|_p^p = O(d^{p})$, whereas for each row $A_i \in \R^d$ of $A$, interpreted as a column vector, $\|AA_i\|_p^p \ge d^p$. 
Intuitively, the lower bound comes from the fact that one can recover the whole matrix $A$ by querying all Boolean vectors $x \in \bcube^d$ using the function $Q_p(\cdot)$, since if $x$ is a row of $A$, then $\|Ax\|_p^p$ would be slightly larger than its typical value, by adjusting constants. 

To implement this idea, one can generate a set of almost orthogonal vectors $S \subseteq \R^d$ and require that all rows of $A$ come from $S$.
A simple probabilistic argument shows that one can construct a set of $|S| = d^p$ vectors such that for any distinct $s, t \in S$, $\left| \langle s, t \rangle \right| \le O(\sqrt{d \log d})$\footnote{The $O(\sqrt{\log d})$ factor can be removed using more sophisticated constructions based on coding theory (see Lemma \ref{lem:selb_hard}).}. 
If we form the matrix $A$ using $n = \widetilde{\Omega}(d^{p / 2})$ vectors from $S$ as its rows, then for any vector $t$ that is {\em not} a row of $A$,
$$
\|At\|_p^p \le n \cdot (d \log d)^{p / 2} \ll d^{p}
$$
for some appropriate choice of $n$.
Thus, by querying $Q_p(s)$ for all vectors $s \in S$, one can recover the whole matrix $A$, even when $\varepsilon$ is a constant. 
By a standard information-theoretic argument, this leads to a lower bound of $\Omega\left(\log \binom{d^{p}}{\widetilde{\Omega}(d^{p / 2})}\right) = \widetilde{\Omega}(d^{p / 2})$ bits.
Furthermore, one only needs to query $|S| = d^p$ vectors, which means the lower bound in fact holds for the for-each version of the $\ell_p$ subspace sketch problem, by a standard repetition argument and losing a $\log d$ factor in the lower bound. 

\paragraph{An $\Omega(d^{\max\{p / 2, 1\} + 1})$ Lower Bound for the For-All Version.}

In order to obtain the nearly optimal $\Omega(d^{\max\{p / 2, 1\} + 1})$ lower bound for the for-all version, we must abandon the constraint that all rows of the $A$ matrix come from a set $S$ of $\poly(d)$ vectors. 
Our plan is still to construct a large set of matrices $\mathcal{S} \subseteq \{+1, -1\}^{\Theta(d^{p / 2}) \times d}$, and show that for any distinct matrices $S, T \in \mathcal{S}$, it is possible to distinguish them using the function $Q_p(\cdot)$, thus proving an $\Omega(\log |\mathcal{S}|)$ lower bound. 
The new observation is that, to distinguish two matrices $S, T \in \mathcal{S}$, it suffices to have a {\em single} row of $T$, say $T_i$, such that $\|ST_i\|_p^p \ll d^p$.
Again using the probabilistic method, we show the existence of such a set $\mathcal{S}$ with size $\exp\left(\Omega(d^{p / 2 + 1}) \right)$, which implies an $\Omega(\log |\mathcal{S}|) = \Omega(d^{p / 2 + 1})$ lower bound.

Our main technical tool is Talagrand's concentration inequality, which shows that for any $p \ge 2$ and vector $x \in \bcube^d$, for a matrix $A \in \R^{\Theta(d^{p / 2}) \times d}$ with \iid Rademacher entries, $\|Ax\|_p=\Theta(d)$ with probability $1 - \exp(-\Omega(d))$. 
This implies that for two random matrices $S, T \in \R^{\Theta(d^{p / 2}) \times d}$ with \iid Rademacher entries, the probability that there exists some row $T_i$ of $T$ such that $\|ST_i\|_p^p \ll d^p$ is at least $1 - \exp\left(\Omega(d^{p / 2 + 1})\right)$, since the $\Theta(d^{p / 2})$ rows of $T$ are independent. 
By a probabilistic argument, the existence of the set $\mathcal{S}$ follows.
The formal analysis is given in Section \ref{sec:p>=2}.

The above argument fails to give an $\Omega(d^2)$ lower bound when $p < 2$. 
However,  for any $p < 2$, since $\ell_2^n$ embeds into $\ell_p^m$ with $m = O_{p}(n)$ and a constant distortion, we can directly reduce the case of $p < 2$ to the case of $p = 2$. The formal analysis can be found in Section \ref{sec:p<2}.
Combining these two results yields the $\Omega(d^{\max\{p / 2, 1\} + 1})$ lower bound.

% !TEX root = main.tex
\section{Preliminaries}
For a vector $x \in \R^n$, we use $\|x\|_p$ to denote its $\ell_p$-norm, i.e., $\|x\|_p = \left(\sum_{i=1}^n |x_i|^p\right)^{1/ p }$. When $p<1$, it is not a norm but it is still a well-defined quantity and we call it the $\ell_p$-norm for convenience. 
When $p = 0$, $\|x\|_0$ is defined to be the number of nonzero coordinates of $x$.

For two vectors $x, y \in \R^n$, we use $\proj_y x \in \mathbb{R}^n$ to denote the orthogonal projection of $x$ onto $y$.
For a matrix $A\in \R^{n\times d}$, we use $A_i \in \R^d$ to denote its $i$-th row, treated as a column vector. 
We use $\|A\|_2 $ to denote its spectral norm, i.e., $\|A\|_2 = \sup_{\|x\|_2 = 1}\|Ax\|_2$, and $\|A\|_F$ to denote its Frobenius norm, i.e., $\|A\|_F = \big(\sum_{i=1}^n \sum_{j= 1}^d A_{ij}^2\big)^{1/2}$.

Suppose that $A \in \mathbb{R}^{m \times n}$ has singular values $\sigma_1 \ge \sigma_2 \ge \cdots \ge \sigma_r \ge 0$, where $r = \min\{m,n\}$. It holds that $\sigma_1 = \|A\|_2 \le \|A\|_F = \big(\sum_{i=1}^r \sigma_i^2\big)^{1/2}$. 
The condition number of $A$ is defined to be
\[
\kappa(A) = \frac{\sup_{\|x\|_2=1} \|Ax\|_2}{\inf_{\|x\|_2=1} \|Ax\|_2}.
\]

\begin{theorem}[Eckart--Young--Mirsky Theorem]\label{thm:svd}
Suppose that $A \in \mathbb{R}^{m \times n}$ has singular values $\sigma_1 \ge \sigma_2 \ge \cdots \ge \sigma_r > 0$, where $\rank(A) = r \le \min\{m,n\}$. For any matrix $B \in \mathbb{R}^{m \times n}$ such that $\rank(B) \le k \le r$, it holds that
$$
\|A - B\|_F^2 \ge \sum_{i = k + 1}^{r} \sigma_i^2.
$$
\end{theorem}

Below we list a handful of concentration inequalities which will be useful in our arguments.

\begin{lemma}[Hoeffding's inequality, {\cite[p34]{BLM}}]\label{lem:hoeffding}
Let $s_1,\dots,s_n$ be \iid Rademacher random variables and $a_1,\dots,a_n$ be real numbers. Then 
\[
\Pr\left\{\sum_i s_ia_i > t\right\}\leq \exp\left(-\frac{t^2}{2\sum_i a_i^2}\right).
\]
\end{lemma}

\begin{lemma}[Khintchine's inequality, {\cite[p145]{BLM}}]\label{lem:khintchine}
Let $s_1,\dots,s_n$ be \iid Rademacher random variables and $a_1,\dots,a_n$ be real numbers. There exist absolute constants $A,B>0$ such that 
\[
A\left(\sum_i a_i^2\right)^{1/2} \leq \left(\E |\sum_i s_i a_i|^p\right)^{1/p} \leq B\sqrt{p}\left(\sum_i a_i^2\right)^{1/2}.
\]
\end{lemma}

\begin{lemma}[Talagrand's inequality, {\cite[p204]{BLM}}]\label{lem:talagrand}
Let $X = (X_1,\dots,X_n)$ be a random vector with independent coordinates taking values in $[-1,1]$. Let $f:[-1,1]^n\to \R$ be a convex $1$-Lipschitz function. It holds for all $t\geq 0$ that 
\[
\Pr\left\{f(X) - \E f(X) \geq t\right\}\leq e^{-t^2/8}.
\]
\end{lemma}

\begin{lemma}[Gaussian concentration, {\cite[p105]{vershynin}}]\label{lem:gaussian_lipschitz}
Let $p\geq 1$ be a constant. Consider a random vector $X\sim N(0,I_n)$ and a non-negative $1$-Lipschitz function $f:(\R^n,\|\cdot\|_2)\to \R$, then 
\[
\Pr\left\{ |f(x) - (\E (f(x))^p)^{1/p} |\geq t \right\}\leq 2e^{-ct^2},
\]
where $c = c(p) > 0$ is a constant that depends only on $p$.
\end{lemma}

\begin{lemma}[Extreme singular values, {\cite[p91]{vershynin}}]\label{lem:singular_values}
Let $A$ be an $N\times n$ matrix with \iid Rademacher entries. Let $\sigma_{\min}(A)$ and $\sigma_{\max}(A)$ be the smallest and largest singular values of $A$. Then for every $t\geq 0$, with probability at least $1-2\exp(ct^2)$, it holds that 
\[
\sqrt{N}-C\sqrt{n}-t\leq \sigma_{\min}(A) \leq \sigma_{\max}(A) \leq  \sqrt{N}+C\sqrt{n}+t,
\]
where $C,c>0$ are absolute constants. 
\end{lemma}

\begin{lemma}\label{lem:r_bound}
Let $U_1,\dots,U_k\in \R^n$ be orthonormal vectors and $s_1,\dots,s_k$ be independent Rademacher random variables. It holds that
$$
\Pr\left\{\left\|  \sum_{i=1}^k s_i \cdot U_i \right\|_{\infty} \le 3\sqrt{\ln k}\right\} \ge 1 - \frac{1}{k^{1.3}}.
$$
\end{lemma}
\begin{proof}
Let $Z = \sum_{i=1}^k s_i U_i$, then
$$
Z_j = \sum_{i=1}^k s_i  U_{i, j}.
$$
Since  $\left\{ U_{i} \right\}$ is a set of orthonormal vectors, we have that
$$
 \sum_{i=1}^k U_{i, j}^2 \le 1.
$$
It follows from Hoeffding's inequality (Lemma~\ref{lem:hoeffding}) that for each $j \in [k]$,
\[
\Pr\{|Z_j|\geq 3\sqrt{\ln k}\} \leq \exp(-2\ln k).
\]
The claimed inequality follows by taking a union bound over all $j \in [k]$.
\end{proof}

We also need a result concerning uniform approximation of smooth functions by polynomials. Let $P_n$ denote the space of  polynomials of degree at most $n$. For a given function $f\in C[a,b]$, the best degree-$n$ approximation error $E_n(f; [a,b])$ is defined to be
\[
E_n(f;[a,b]) = \inf_{p\in P_n} \|f-p\|_\infty,
\]
where the $\|\cdot\|_\infty$ norm is taken over $[a,b]$. The following bound on approximation error is a classical result.

\begin{lemma}[{\cite[p23]{rivlin}}]\label{lem:poly_approx}
Let $f(x)$ have a $k$-th derivative on $[-1,1]$. If $n>k$, 
\[
	E_n(f;[-1,1]) \leq \frac{6^{k+1}e^k}{(k+1)n^k}\ \omega_k\left(\frac{1}{n-k}\right),
\]
where $\omega_k$ is the modulus of continuity of $f^{(k)}$, defined as
\[
	\omega_k(\delta) = \sup_{\substack{x,y\in [-1,1]\\ |x-y| \leq \delta}} \left|f^{(k)}(x) - f^{(k)}(y)\right|.
\]
\end{lemma}

% !TEX root = main.tex
\section{An $\widetilde{\Omega}\left(\eps^{-2}\right)$ Lower Bound}\label{sec:1/eps^2 lb}
\newcommand{\midsize}{N^{(d)}}
\newcommand{\midvalue}{\Lambda^{(d, p)}_0}
\newcommand{\diag}{\Lambda^{(d, p)}}
\newcommand{\had}{H^{(d)}}
\newcommand{\improw}{\mathcal{R}}
\newcommand{\shiftr} {\Delta^{(d)}}
\newcommand{\ubf}{\overline{F}^{(d, p)}}

To prove the space lower bound of the data structure $Q_p$, we appeal to information theory. We shall encode random bits in $A$ such that if for each $x$, $Q_p(x)$ approximates $\|Ax\|_p^p$ (or $\|Ax\|_p$ when $p=0$) up to a $1\pm\eps$ factor with probability at least $0.9$, we can recover from $Q_p(x)$ some random bit (depending on $x$) with at least constant probability. A standard information-theoretic argument implies a lower bound on the size of $Q_p$ which is proportional to the number of random bits we can recover.

For each $p \ge 0$, we define a family of matrices $M^{(p)} = \{M^{(d, p)}\}_{d = 1}^{\infty}$, where $M^{(d, p)}$ is a $2^d \times 2^d$ matrix with entries defined as
$$
M^{(d, p)}_{i, j} = |\langle i, j \rangle|^p,
$$
where $i$ and $j$ are interpreted as vectors on the Boolean cube $\bcube^d$. We assume $0^0 = 0$ throughout the paper.

% !TEX root = main.tex

\subsection{Spectrum of Matrices $M$}\label{sec:spec}
\begin{lemma}\label{lem:decomp}
For any $d \ge 1$, $M^{(d, p)}$ can be rewritten as $\had \diag (\had)^T$ in its spectral decomposition form, where $\diag$ is a $2^d \times 2^d$ diagonal matrix, and $\had$ is a $2^d \times 2^d$ normalized Hadamard matrix. 
\end{lemma}
\begin{proof}
Let $T$ be the natural isomorphism from the multiplicative group $\bcube^d$ to the additive group $\F_2^d$. Then $|\langle i,j\rangle|^p = g(Ti+Tj)$ for some function $g$ defined on $\F_2^d$. It can be computed (see~\cite[Lemma 5]{disc14}) that the singular values of a matrix with entries $g(Ti+Tj)$ are the absolute values of the Fourier coefficients of $g$. In our particular case, the singular values of $M^{(d, p)}$ are
\[
\left|\hat g(s)\right| = \left|\sum_{x\in \F_2^d} (-1)^{\langle s,x\rangle} g(x)\right|,\quad s\in \F_2^d.
\]
Furthermore, the proof of that lemma shows that $\had$ in the spectral decomposition is given by
\[
(\had_s)_z = \frac{1}{2^{d/2}}(-1)^{\langle s,z\rangle},\quad s,z\in \F_2^d,
\] 
which implies that $\had$ is a normalized Hadamard matrix.
\end{proof}

\begin{lemma}\label{lem:fourier_coeff}
When $d$ is even, there are at least $\binom{d}{d/2}$ entries in $\diag$ with absolute value 
\[
\left| \sum_{\substack{0\leq i\leq d\\ \text{$i$ is even}}} (-1)^{i/2} \binom{d / 2}{i / 2}  |d - 2i|^p\right| \triangleq  \midvalue \ge 0.
\]
\end{lemma}
\begin{proof}
We shall use the notation in the proof of Lemma~\ref{lem:decomp}. Consider the Fourier coefficients $\hat g(s)$ for $s\in \F_2^d$ with Hamming weight $d/2$, which is the same for all $\binom{d}{d/2}$ such $s$'s. Note that
\[
\hat g(s) = \sum_{i=0}^{d} \sum_{j=0}^{i} (-1)^j \binom{d / 2}{j}\binom{d / 2}{i-j} g(i).
\]
By comparing the coefficients of $x^i$ on both sides of the identity $(1+x)^{d/2}(1-x)^{d/2} = (1-x^2)^{d/2}$, we see that
\begin{equation*} 
\sum_{j=0}^{i} (-1)^j \binom{d / 2}{j}\binom{d / 2}{i-j} = \begin{cases}																										(-1)^{i/2}\binom{d / 2}{i/2}, & i\text{ is even};\\
0, & i\text{ is odd}.																									\end{cases}
\end{equation*}
Hence
\[
\hat g(s) = \sum_{\text{even }i} (-1)^{i/2}\binom{d / 2}{i/2} g(i).
\]
Finally, observe that, for $x,y\in \{+1,-1\}^d$, we have $(d-\langle x,y\rangle)/2 = d_H(x,y) = w_H(Tx + Ty)$, where $d_H(x,y)$ denotes the Hamming distance between $x$ and $y$ and $w_H(s)$ denotes the Hamming weight of $s\in \F_2^d$. Hence $g(i) = |d-2i|^p$ and the conclusion follows.
\end{proof}

Let $\midsize$ be the multiplicity of the singular value $\midvalue$ of $M^{(d, p)}$. We know from the preceding lemma that $\midsize\geq \binom{d}{d/2}$. By permuting the columns of $\had$, we may assume the absolute value of the first $\midsize$ diagonal entries of $\diag$ are all equal to $\Lambda^{(d, p)}_0$, i.e., 
\[
\left|\diag_1\right| = \left|\diag_2\right| = \cdots = \left|\diag_{\midsize}\right| = \midvalue.
\]

The following lemma is critical in lower bounding $\midvalue$.
We found the result in a post on \texttt{math.stackexchange.com}~\cite{stackexchange} but could not find it in any published literature and so we reproduce the proof in full from~\cite{stackexchange}, with small corrections regarding convergence of integrals.

\begin{lemma}\label{lem:critical_identity} It holds for all complex $p$ satisfying $0<\Re p<2n$ that
\begin{equation}\label{eqn:critical_identity}
\sum_{k=1}^n (-1)^{k+1} \binom{2n}{n+k} k^p =
2^{2n-p}\frac{\Gamma(p+1)}{\pi} \left(\sin\frac{\pi p}{2}\right) \int_0^\infty \frac{\sin^{2n}t}{t^{p+1}} dt.
\end{equation}
\end{lemma}
\begin{proof}
By the binomial theorem,
\[
(z-1)^{2n} = \sum_{k=0}^{2n}\binom{2n}{k}(-z)^k =\sum_{k=-n}^n\binom{2n}{k+n}(-z)^{k+n}.
\]
Splitting the sum at $k=-1$ and $k=1$, we have
$$\sum_{k=1}^n (-1)^k\binom{2n}{n+k}\big(z^k + z^{-k} \big) = (-1)^n(z-1)^{2n}\,z^{-n} - \binom{2n}{n}.$$
Plugging in $z=\exp{(2it)}$ yields
\begin{equation}\label{eqn:aux}
(2\sin t)^{2n}= 2\sum_{k=1}^n (-1)^k\binom{2n}{n+k}\cos(2kt)
+ \binom{2n}{n}.
\end{equation}
Plug \eqref{eqn:aux} into the integral on the right-hand side of \eqref{eqn:critical_identity} and introduce a regularizer $\exp(-st)$ ($s>0$) under the integral sign:
\[
\int_0^\infty e^{-st} \frac{(2\sin{t})^{2n}}{t^{p+1}}dt = 
 2\sum_{k=1}^n (-1)^k\binom{2n}{n+k} \int_0^\infty \frac{e^{-st}\cos(2kt)}{t^{p+1}} dt
+ \binom{2n}{n} \Gamma(-p)s^{p},\quad -1<\Re p<0.
\]
One can compute that
\begin{multline*}
\int_0^\infty \frac{e^{-st}\cos(2kt)}{t^{p+1}}dt =
\frac{1}{2}\int_0^\infty \frac{e^{-st}(e^{i2kt}+e^{-i2kt})}{t^{p+1}}dt = \frac{(s-2ki)^p+(s+2ki)^p}{2}\Gamma(-p)\\
= (4k^2+s^2)^{\frac{p}{2}}\cos\left(p\arctan\frac{2k}{s}\right)\Gamma(-p),\quad -1<\Re p<0.
\end{multline*}
It follows that
\begin{multline*}
\int_0^\infty e^{-st} \frac{(2\sin{t})^{2n}}{t^{p+1}}dt \\
= 
 2\sum_{k=1}^n (-1)^k\binom{2n}{n+k}(4k^2+s^2)^{\frac p2}\cos\left(p\arctan\frac{2k}{s}\right)\Gamma(-p)
+ \binom{2n}{n} \Gamma(-p)s^{p},\quad -1<\Re p<0.
\end{multline*}
It is easy to verify that the integral on the left-hand side is analytic whenever the integral converges. Analytic continuation permits $p$ to be extended to $\{p: -1<\Re p<2n\}\setminus \Z$. Now, for $p$ such that $0<\Re p<2n$ and $p\notin \Z$, let $s\to 0^+$ on both sides. It is also easy to verify that we can take the limit $s\to 0^+$ under the integral sign, hence
\begin{equation}\label{eqn:aux2}
\int_0^\infty \frac{(2\sin{t})^{2n}}{t^{p+1}}dt = 2\sum_{k=1}^n (-1)^k\binom{2n}{n+k} (2k)^p\cos\left(\frac{\pi p}{2}\right)\Gamma(-p),\quad 0<\Re p<2n, p\not\in \Z.
\end{equation}
Invoking the reflection identity (see, e.g.~\cite[p9]{special_functions})
\begin{equation}\label{eqn:reflection_identity}
\Gamma(-p)\Gamma(1+p) = -\frac{\pi}{\sin(p\pi)},\quad p\not\in\Z
\end{equation}
we obtain that
\[
\sum_{k=1}^n (-1)^k\binom{2n}{n+k} k^p = -\frac{2^{2n-p}\Gamma(p+1)}{\pi}\left(\sin\frac{p\pi}{2}\right)\int_0^\infty \frac{\sin^{2n} t}{t^{p+1}}dt,\quad 0<\Re p<2n, p\not\in \Z.
\]
Finally, analytic continuation extends $p$ to the integers in $(0,2n)$.
\end{proof}

As an immediate corollary of Lemma~\ref{lem:critical_identity}, we have:
\begin{corollary}\label{cor:singular_value_lb} Suppose that $d \in 8\Z$. There exists an absolute constant $c > 0$ such that
\[
\midvalue \geq c\frac{2^{d/2}}{\sqrt d} \left|\sin\frac{p \pi}{2}\right|.
\]
\end{corollary}
\begin{proof}
Letting $2n=d/2$ and $k=n-i/2$, the summation in Lemma~\ref{lem:fourier_coeff} becomes
\[
2^{2p+1}\left|\sum_{k=1}^n (-1)^k\binom{2n}{n+k} k^p\right| = \frac{2^{d/2}2^{p+1}\Gamma(p+1)}{\pi}\left|\sin\frac{p\pi}{2}\right|\int_0^\infty \frac{\sin^{d} t}{t^{p+1}}dt.
\]
Since (cf. \cite[p511]{courant})
\[
\int_0^\pi \sin^d x\,dx = \frac{\sqrt{\pi}\Gamma(\frac{d+1}{2})}{\Gamma(\frac d2)}\geq \frac{C}{\sqrt d},
\]
where $C>0$ is an absolute constant, we have that
\[
\int_0^\infty \frac{\sin^{d} t}{t^{p+1}}dt \geq \sum_{n=0}^\infty \frac{1}{((n+1)\pi)^{p+1}} \int_{n\pi}^{(n+1)\pi} \sin^d x\,dx \geq \frac{C}{\sqrt d}\cdot\frac{\zeta(p+1)}{\pi^{p+1}}.
\]
Notice that $h(p) = \Gamma(p+1)\zeta(p+1)/(\pi/2)^{p+1}$ is a positive continuous function on $(0,\infty)$ and $h(p)\to\infty$ as $p\to\infty$ and $p\to 0^+$, it must hold that $\inf_{p>0} h(p) > 0$. The conclusion follows.
\end{proof}

% !TEX root = main.tex
\subsection{Orthogonalizing Rows}\label{sec:or}
Suppose we are given a matrix $\Pi \in \mathbb{R}^{n \times n}$ in its spectral decomposition form $\Pi = H \Sigma H^T$, where
$$
\Sigma_{i, i} = \begin{cases}
\pm \sigma, & i \le r; \\
0, & r < i \le n,
\end{cases}
$$
and $H$ is the normalized Hadamard matrix.
The goal of this section is to identify a set of orthogonal vectors, using rows of $\Pi$.

\begin{lemma}\label{lem:row_norm}
Each row of $\Pi$ has the same $\ell_2$ norm $\|\Pi_i\|_2= \sigma \sqrt{r / n}$.
\end{lemma}
\begin{proof}
$$
\|\Pi_i\|_2 = \|H_i \Sigma H^T\|_2 = \|H_i \Sigma\|_2.
$$
The lemma follows since all entries in $H$ have absolute value $1 / \sqrt{n}$, and the $r$ non-zero entries on the diagonal of $\Sigma$ have absolute value $\sigma$.
\end{proof}

To identify a set of orthogonal vectors using the rows of $\Pi$, we run a procedure similar to the standard Gram-Schmidt process. 
\begin{lemma}\label{lem:gram_schmidt}
There is a set $\improw \subseteq [n]$ with size $|\improw| = r / 100$ such that for each $i \in \improw$, $\Pi_i$ can be written as
\begin{equation}\label{equ:Pi_decomp}
\Pi_i = R_{i} + P_{i},
\end{equation}
where $\left\{ R_{i} \right\}_{i \in \improw}$ is a set of orthogonal vectors, and $P_{i}$ is the orthogonal projection of $\Pi_i$ onto the subspace spanned by $\left \{ R_{j} \right\}_{j \in \improw \setminus \{i\}}$.
Furthermore, for each $i \in \improw$, $\|R_{i}\|_2^2 \ge 99 / 100 \|\Pi_i\|_2^2 = 99/100 \cdot \sigma^2 r / n$.
\end{lemma}
\begin{proof}
We show how to construct such a set $\improw$.
Suppose that we have found a set $\improw$ with size strictly less than $r / 100$ with $\Pi_i = R_i + P_i$ satisfying the stated constraints. We shall show how to increase the size of $\improw$ by one.

Let $\Pi = S + Q$. Here, for each $i \in [n]$ we have $\Pi_i = S_{i} + Q_{i}$, where $Q_{i}$ is the orthogonal projection of $\Pi_i$ onto the subspace spanned by $\left \{ R_{j} \right\}_{j \in \improw}$ and $S_i = \Pi_i - Q_i$.
Notice that for all $j \in \improw$ we have $Q_j = \Pi_j$ and $S_j = 0$.
Since $\|\Pi\|_F^2 = r  \sigma ^2$ and $\rank(Q) \le |\improw|$, by Theorem \ref{thm:svd} we have 
\[
\sum_{i \in [n]} \|S_i\|_2^2  = \|S\|_F^2 = \|\Pi - Q\|_F^2 \ge \|\Pi\|_F^2 - |\improw| \cdot  \sigma ^2 > \frac{99}{100} \|\Pi\|_F^2 = \frac{99}{100}\sum_{i \in [n]} \|\Pi_i\|_2^2 .
\]
Thus, by averaging, there exists $i \notin \improw$ such that $ \|S_{i}\|_2^2  > 99 / 100 \| \Pi_i\|_2^2$. 
We add $i$ into $\improw$ and set $R_i$ in \eqref{equ:Pi_decomp} to be $S_i$ and $P_i$ to be $Q_i$.
It is easy to verify that the stated constraints still hold. 
We continue this process inductively until $|\improw| = r / 100$.
% It holds that $\|R_{i}\|_2^2 = \|\overline{R}_{i}\|_2^2 \ge 99 / 100 \|\Pi_i\|_2^2 $. 
\end{proof}

\begin{lemma}\label{lem:small_error}
Suppose that $e \in \mathbb{R}^{n}$ satisfies $\|e\|_{\infty} \le 0.1 \sigma \sqrt{r / n}$. 
Let $x \in \mathbb{R}^n$ be a random vector defined as
$$
x = \sum_{i \in \improw} s_i \cdot \frac{R_i}{\|R_i\|_2},
$$
where $\{s_i\}_{i \in \improw}$ is a set of \iid Rademacher random variables. 
Here the set $\improw$ and the orthogonal vectors $\{R_i\}_{i \in \improw}$ are as defined in Lemma \ref{lem:gram_schmidt}.
For each $i \in \improw$, it holds that
$$
\Pr_x\left\{\sign\left( (\Pi x+e)_i \right) = \sign(s_i)\right\}\geq \frac{4}{5}.
$$
\end{lemma}
\begin{proof}
For each $i \in \improw$, we have
$$
 \langle \Pi_i, x\rangle =  \langle R_{i}, x\rangle +  \langle P_{i}, x\rangle = s_i \cdot \|R_i\|_2 + \sum_{j \in \improw \setminus\{i\}} s_j \cdot \|\proj_{R_j} \Pi_i\|_2.
$$
We first analyze the second term. 
$$
\E \left| \langle P_{i}, x\rangle \right| \leq  \left(\sum_{j \in \improw \setminus \{i\}} \|\proj_{R_j} \Pi_i\|_2^2 \right)^{1/2} = \| P_{i} \|_2 \le  \frac{1}{10}\|\Pi_i\|_2.
$$
By Markov's inequality, with probability at least $4/5$, we have $\left| \langle P_{i}, x\rangle \right| \le  \|\Pi_i\|_2 / 2$. 

Recall that $\|R_{i}\|_2 \ge 99 / 100 \|\Pi_i\|_2$ (Lemma~\ref{lem:gram_schmidt}) and 
$\|\Pi_i\|_2 = \sigma \sqrt{r / n}$ (Lemma~\ref{lem:row_norm}). It happens with probability at least $4/5$ that $|e_i| +  \left|  \langle P_{i}, x\rangle \right| <  \left|   \langle R_{i}, x\rangle  \right|$, in which case we have $\sign\left( (\Pi x+e)_i \right) = \sign(s_i)$.
\end{proof}

\subsection{Space Lower Bound on $Q_p$}\label{sec:comm_lb}
In this section, we describe a reduction from the subspace sketch problem to the $\mathsf{INDEX}$ problem, a classical problem in communication complexity. We shall rephrase the problem in the context of a data structure. The $\mathsf{INDEX}$ data structure stores an input string $s \in \bcube^n$ and supports a query function, which receives an input $i \in [n]$ and outputs $s_i \in \{-1, 1\}$ which is the $i$-th bit of the underlying string. To prove the lower bound for the subspace sketch problem, we need the following lower bound for the distributional $\mathsf{INDEX}$ problem.

\begin{lemma}[\cite{miltersen1998data}]\label{lem:lb_game}
In the $\mathsf{INDEX}$ problem, suppose that the underlying string $s$ is drawn uniformly from $\bcube^n$ and the input $i$ of the query function is drawn uniformly from $[n]$. Any (randomized) data structure for $\mathsf{INDEX}$ that succeeds with probability at least $2 / 3$ requires $\Omega(n)$ bits of space, where the randomness is taken over both the randomness in the data structure and the randomness of $s$ and $i$.
\end{lemma}

Throughout the reduction, $d$ is a fixed parameter with value to be determined later.
For the matrix $M^{(d, p)}$, we consider its spectrum-truncated version \[
\tilde{M}^{(d, p)} \triangleq \had \mathrm{diag}(\diag_1, \diag_2. \ldots, \diag_{\midsize}, 0, 0, \ldots, 0) (\had)^T.
\]

\begin{lemma}\label{lem:ortho_with_A}
Each row of $\tilde{M}^{(d, p)}$ is orthogonal to all eigenvectors associated with eigenvalues other than $\diag_1, \diag_2, \ldots, \diag_{\midsize}$.
\end{lemma}
\begin{proof}
Let $v_1,\dots,v_{2^d}$ be the columns of $\had$. Then $\tilde M^{(d, p)} = \sum_{i=1}^{\midsize} \Lambda_i^{(d, p)} v_iv_i^T$. Let $w$ be an eigenvector corresponding to another eigenvalue. Then
\[
\tilde M^{(d, p)}w = \sum_{i=1}^{\midsize} \Lambda_i^{(d, p)} v_i (v_i^Tw) = 0,
\]
since $v_i$ and $w$ are orthogonal as they are associated with distinct eigenvalues.
\end{proof}
Now we invoke Lemma \ref{lem:gram_schmidt} on the matrix $\tilde{M}^{(d, p)}$ and obtain a set $\improw \subseteq [2^d]$ and a set of orthogonal vectors $\{R_i\}_{i \in \improw}$. We shall encode $|\improw|$ random bits in $A$ and show how to recover them.
%Notice that Alice and Bob can locally compute the matrix ${\tilde A}^{(d, p)}$, and thus the set $\improw$ and the set of orthogonal vectors $\{R_i\}_{i \in \improw}$.
%In our reduction to the distributional $\mathsf{INDEX}$ problem, we set $n = |\improw|$.

% Upon finishing drawing the \iid Rademacher random variables $\{s_i\}_{i \in \improw}$
%, which are Alice's inputs to the \textsf{INDEX} problem, Alice calculates
Let
\[
x = \sum_{i \in \improw} s_i \cdot \frac{R_i}{\|R_i\|_2}.
\]
By Lemma \ref{lem:r_bound}, with probability $1 - \exp(-\Omega(d))$, it holds that $\|x\|_{\infty} \le 3\sqrt{d}$. % If 
We condition on $\|x\|_{\infty} \le 3\sqrt{d}$ in the rest the proof, since we can include the alternative case $\|x\|_\infty > 3\sqrt{d}$ in the overall failure probability.

Next we define a vector $y\in \R^{2^d}$ to be $y_i = ( x_i + \shiftr)^{1/p}$, where $\shiftr = 5\sqrt{d}$ is a constant that depends only on $d$.
Clearly, it holds for all $i \in [2^d]$ that $2\sqrt{d} \le y_i^p \le 8 \sqrt{d}$. Round each entry of $y$ to its nearest integer multiple of $\delta = 1/(p(8\sqrt{d})^{1-1/p}2^d)$, obtaining $\tilde{y}$. 
A simple calculation using the mean-value theorem shows that for all $i \in [2^d]$, 
\begin{equation}\label{equ:error_tilde_x}
|\tilde{y}_i^p -  ( x_i + \shiftr)| = |\tilde{y}_i^p - y_i^p| \leq p(8\sqrt{d})^{\frac{p-1}{p}}\delta \le 2^{-d}. 
\end{equation}

Finally we construct the matrix $A \in \mathbb{R}^{2^d \times d}$ to be used in the $\ell_p$ subspace sketch problem. 
The $j$-th row of $A$ is the $j$-th vector of $\bcube^d$, scaled by $\tilde{y}_j$.

\begin{lemma} The matrix $A$ constructed above for the $\ell_p$ subspace sketch problem satisfies $\kappa(A)\leq C$ for some constant $C$ that depends on $p$ only.
\end{lemma}
\begin{proof}
Let $B$ be the $2^d\times d$ matrix whose rows are all vectors in $\bcube^d$. Then, 
\[
\|Bx\|_2^2 = 2^d \E\left|\sum_{i=1}^d s_i x_i\right|^2,
\] 
where $s_1,\dots,s_d$ is a Rademacher sequence. It follows from Khintchine's inequality that 
\begin{equation}\label{eqn:khintchine_B}
C_1 2^{d/2}\|x\|_2 \leq \|Bx\|_2 \leq C_2 2^{d/2}\|x\|_2
\end{equation}
for some constants $C_1,C_2$.
Notice that the rows of $A$ are rescaled rows of $B$ with the scaling factors in $[(2\sqrt{d})^{1/p},(8\sqrt{d})^{1/p}]$. Hence $\kappa(A)\leq C$
for some constant $C$ that depends on $p$ only.
\end{proof}

The recovery algorithm is simple. The vector to be used for querying the data structure is the $i$-th vector on the Boolean cube $\bcube^d$, where $i \in \improw$. Given $Q_p(i)$, we guess the sign of $s_i$ to be just the sign of $Q_p(i) - \langle M^{(d, p)}_{i}, \shiftr \cdot \mathbf{1} \rangle$. 
Next we prove the correctness of the recovery algorithm.

The guarantee of the subspace sketch problem states that, with probability at least $0.9$, it holds simultaneously for all $i \in \bcube^d$ that
\begin{equation}\label{eqn:query_guarantee}
\|Ai\|_p^p \le Q_p(i) \le (1 + \varepsilon) \|Ai\|_p^p.
\end{equation}
We condition on this event in the remaining part of the analysis.

First we notice that for any $i \in \bcube^d$,
\begin{equation}\label{equ:exp_Ai}
\|Ai\|_p^p = \sum_{j = 1}^{2^d} |\langle A_{j}, i \rangle|^p = \sum_{j \in \bcube^d}  |\langle i,\tilde{y}_j \cdot  j\rangle|^p = \sum_{j \in \bcube^d}  \tilde{y}_j^p |\langle i,   j\rangle|^p .
\end{equation}

Next we give an upper bound on the value of $\|A  i\|_p^p$, for all $i \in \bcube^d$.
\begin{lemma}
For each $i \in \bcube^d$, 
the matrix $A$ constructed for the $\ell_p$ subspace sketch problem satisfies
$$
\|Ai\|_p^p \le 2^d \cdot (8d^{1.5})^p.
$$
\end{lemma}
\begin{proof}
Each term in the summation \eqref{equ:exp_Ai}
is upper bounded by $(8d^{1.5})^p$,  which implies the stated lemma. 
\end{proof}
Combining the preceding lemma with the query guarantee~\eqref{eqn:query_guarantee}, the preceding lemma implies that, it holds for all $i \in \mathbb{R}^{2^d}$ that 
$$
|Q_p(i) - \|Ai\|_p^p| \le \varepsilon \cdot  2^d \cdot (8d^{1.5})^p.
$$

On the other hand, by \eqref{equ:error_tilde_x} and \eqref{equ:exp_Ai},
$$
 \left| \|Ai\|_p^p - \langle M^{(d, p)}_{i}, (x + \mathbf{1} \cdot \shiftr) \rangle \right| \le \sum_{j \in \bcube^d} |\langle i, j \rangle|^p \cdot  |\tilde{y}_i^p -  ( x_i + \shiftr)| \le d^{p}.
$$ 
Thus by the triangle inequality, 
$$
\left| (Q_p(i) - \langle M^{(d, p)}_{i}, \shiftr \cdot \mathbf{1} \rangle ) - ( \langle M^{(d, p)}_{i}, x\rangle)  \right| \le \varepsilon \cdot  2^d \cdot (8d^{1.5})^p + d^p.
$$
Notice that $x$ is a linear combination of rows of $\tilde{M}^{(d, p)}$. By Lemma~\ref{lem:ortho_with_A}, 
$$
M^{(d, p)} x = \tilde{M}^{(d, p)} x.
$$

By Lemma~\ref{lem:small_error}, if
\begin{equation}\label{equ:calc_d}
\varepsilon \cdot  2^d \cdot (8d^{1.5})^p + d^p  \le  0.1 \midvalue \sqrt{\midsize} / 2^{d/2},
\end{equation}
then with probability $4/5$, $(Q_p(i) - \langle M^{(d, p)}_{i}, \shiftr \cdot \mathbf{1} \rangle )$ has the same sign as $\left( \tilde{M}^{(d, p)} x \right)_i$, in which case we recover the correct sign. 
By Lemma~\ref{lem:lb_game}, the size of $Q_p$ is lower bounded by $\Omega(|\improw|)$. 

Now for each $\varepsilon > 0$ and $p\in (0,\infty)\setminus 2\Z$, by Lemma \ref{lem:fourier_coeff}, \eqref{equ:calc_d} can be satisfied by setting
$$
2^{d/2} \geq \frac{\sin(p \pi / 2)}{\varepsilon \cdot  \polylog(1 / \varepsilon)},
$$
which implies a space complexity lower bound of
$$
\Omega(|\improw|) = \Omega(\midsize) = \Omega(2^d / \sqrt{d}) = \Omega \left( \frac{1}{\varepsilon^2 \cdot  \polylog(1 / \varepsilon)}\right)
$$
bits.

Formally, we have proved the following theorem.
\begin{theorem}\label{thm:comm_lb_eps}
Let $p\in (0,\infty)\setminus 2\Z$. There exist constants $C \in (0,1]$ and $\eps_0 > 0$ that depend only on $p$ such that the following holds. Let $d_0 = 2\log_2(C/(\eps\polylog(1/\eps))$. For any $\eps \in (0,\eps_0)$, $d \geq d_0$ and $n\geq 2^{d_0}$, any data structure for the $\ell_p$ subspace sketch requires $\Omega \left( \frac{1}{\varepsilon^2 \cdot  \polylog(1 / \varepsilon)}\right)$ bits.
The lower bound holds even when $\kappa(A)\leq K$ for some constant $K$ that only depends on $p$.
\end{theorem}

We note that the $\polylog(1/\eps)$ factors in the definition of $d_0$ and the bit lower bound may not have the same exponent.

Next we strengthen the lower bound to $\widetilde{\Omega}(d/\varepsilon^2)$ bits.

\begin{corollary}\label{cor:comm_lb_eps}
Under the assumptions of $C$, $\epsilon_0$, $d$ in Theorem~\ref{thm:comm_lb_eps} and the assumption that $n = \Omega \left( \frac{d}{\varepsilon^2 \cdot  \polylog(1 / \varepsilon)}\right)$, 
any data structure for the $\ell_p$ subspace sketch problem requires $\Omega \left( \frac{d}{\varepsilon^2 \cdot  \polylog(1 / \varepsilon)}\right)$ bits. The $\polylog(1/\eps)$ factors in the two $\Omega$-notations may not have the same exponent.
\end{corollary}
\begin{proof}
Let $A'\in \R^{n'\times d'}$ be the hard instance matrix for Theorem~\ref{thm:comm_lb_eps}, where $d' = 2\log_2(C/(\eps\polylog(1/\eps))$ and $n' = 2^{d'}$. We construct a block diagonal matrix $A$ with $b=d/d'$ blocks, each being an independent copy of $A$', so that $A$ has $d$ columns. The number of rows in $A$ is $bn' = \Omega\left(\frac{d}{\eps^2\polylog(1/\eps)}\right)$. In this case, the $\ell_p$ sketch problem on $A'$ requires a data structure of $\widetilde{\Omega}(b/\eps^2) = \Omega \left( \frac{d}{\varepsilon^2 \cdot  \polylog(1 / \varepsilon)}\right)$ bits, since we are now solving the \textsf{INDEX} problem with $\widetilde\Omega(b\cdot 1/\eps^2)$ random bits.
\end{proof}

The corollary above is also true for $p=0$.
\begin{corollary}\label{cor:comm_lb_eps_p=0}
Under the assumptions of $C$, $\epsilon_0$, $d$ and $n$ in Corollary~\ref{cor:comm_lb_eps}, any data structure for the $\ell_0$ subspace sketch problem requires $\Omega \left( \frac{d}{\varepsilon^2 \cdot  \polylog(1 / \varepsilon)}\right)$ bits.
\end{corollary}
\begin{proof}
The matrix $M^{(d,0)}$ is defined as $(M^{(d,0)})_{i,j} = \mathbf{1}_{\{\langle i,j\rangle\neq 0\}}$. Note that each row of $M^{(d,0)}$ has the same number of $1$s; let $W_d$ denote this number. Observe that Lemma~\ref{cor:singular_value_lb} continues to hold because we have by symmetry
\[
\sum_{k=1}^n (-1)^{k+1} \binom{2n}{n+k} = \frac{1}{2}\binom{2n}{n}\geq c\frac{2^{2n}}{\sqrt n}
\]
for some absolute constant $c > 0$. Let $y_j = x_i + \Delta^{(d)}$, where $x_i$ and $\Delta^{(d)}$ are as defined before. In the construction of $A$, replicate $\tilde y_j$ times (rounded to an integer multiple of $\delta=2^{-d}$) the $j$-th vector of $\bcube^d$. Our guess of the sign $s_i$ is then the sign of $\delta Q_0(i) - \Delta^{(d)} W_d$. Similar to the procedure above, we have that
\[
\delta\left|Q_0(i) - \|Ai\|_0\right| \leq \delta \eps \|Ai\|_0 \leq \eps \cdot 8\sqrt d\cdot 2^d
\]
and
\[
\left|\delta\|Ai\|_0 - \langle M^{(d, 0)}_{i}, (x + \mathbf{1} \cdot \shiftr)\rangle\right| \leq \sum_j |\hat y_j - x_i - \Delta^{(d)}| \mathbf{1}_{\{\langle i,j\rangle\neq 0\}} \leq \delta 2^d = 1.
\]
And therefore it suffices to have
\[
\epsilon\cdot 8\sqrt{d}\cdot 2^d + 1 \leq 0.1 \midvalue \sqrt{\midsize} / 2^{d/2},
\]
which holds when $2^{d/2} = 1/(\eps/\polylog(1/\eps))$ as before. Therefore the analogue of Theorem~\ref{thm:comm_lb_eps} holds and so does the analogue of Corollary~\ref{cor:comm_lb_eps}.
\end{proof}

\begin{remark}\label{rem:eps=0_even d}
The condition that $p \notin 2\mathbb{Z}^+$ is necessary for the lower bound.
When $p\in 2\Z^+$, it is possible to achieve $\eps = 0$  with $O(d^p\log(nd))$ words. Recall that a $d$-dimensional subspace of $\ell_p$ space can be isometrically embedded into $\ell_p^r$ with  $r = \binom{d+p-1}{p}-1$~\cite{handbook:21}. In general the data structure does not necessarily correspond to a linear map and can be of any form. Indeed, there is a much simpler data structure as follows, based on ideas in~\cite{schechtman:tight}. For each $x\in\R^d$, let $y_x\in \R^d$ be defined as $(y_x)_i = ((Ax)_i)^{p/2}$, then $\|y_x\|_2^2 = \|Ax\|_p^p$. Observe that each coordinate $(y_x)_i$ is a polynomial of $d^{p/2}$ terms in $x_1,\dots,x_d$. Form an $n\times d^{p/2}$ matrix $B$, where the $i$-th row consists of the coefficients in the polynomial corresponding to $(y_x)_i$. The data structure stores $B^TB$. To answer the query $Q_p(x)$, one first calculates from $x$ a $d^{p/2}$-dimensional vector $x'$ whose coordinates are all possible monomials of total degree $p/2$. Note that $Bx' = y_x$. Hence one can just answer $Q_p(x) = (x')^T B^T Bx' =  \|Bx'\|_2^2 = \|Ax\|_p^p$ without error. This $Q_p$ does not give an isometric embedding but is much simpler than known isometric embeddings, and the space complexity is $O(d^p\log(nd))$ bits.
\end{remark}

% !TEX root = main.tex
\section{Lower Bounds for $p > 2$} \label{lem:d^{p/2}_lb}
\subsection{Lower Bounds for the Subspace Sketch Problem for $p > 2$}
In this section, we prove a lower bound on the $\ell_p$ subspace sketch problem, in the case that $\eps$ is a constant and $p \ge 2$. 
 We need the following result from coding theory.
\begin{lemma}[\cite{parampalli2013construction}]\label{lem:selb_hard}
For any $p \ge 1$ and $d = 2^k - 1$ for some integer $k$, there exist a set $S \subset \bcube^d$ and a constant $C_p$ depending only on $p$ which satisfy
\begin{enumerate}[topsep=0.5ex,itemsep=-0.5ex,partopsep=1ex,parsep=1ex,label=(\roman*)]
\item $|S| = d^p$;
\item For any $s, t \in S$ such that $s \neq t$, $|\langle s, t\rangle| \le C_p\sqrt{d}$.
\end{enumerate}
\end{lemma}

\begin{lemma}\label{lem:selb_hard_set}
For any $p \ge 1, C \ge 1$ and $d = 2^k - 1$ for an integer $k$, there exist a set $S \subset \bcube^d$ with size $|S| = d^{p}$, a set $\mathcal{M} \subset \mathbb{R}^{R \times d}$ for some $R$ and a constant $C_p$ depending only on $p$ which satisfy
\begin{enumerate}[topsep=0.5ex,itemsep=-0.5ex,partopsep=1ex,parsep=1ex,label=(\roman*)]
\item For any $M_1, M_2 \in \mathcal{M}$ such that $M_1 \neq M_2$, there exists $x \in S$ such that $\|Mx\|_p < d / C$ and $\|Mx\|_p \ge d$.
\item $|\mathcal{M}| \ge \exp \left(d^{p / 2} / (C_p C^p) \right)$. 
\end{enumerate}
\end{lemma}
\begin{proof}
Set $R = d^{p / 2} / (C_p C^{p})$. Then $R\leq d^p/e$. We set $\mathcal{M}$ to be the set of $R \times d$ matrices whose rows are all possible combinations of $R$ distinct vectors in $S$, where $S$ is the set constructed in Lemma \ref{lem:selb_hard}.
Clearly, $|\mathcal{M}| = {d^p \choose R} \ge e^{R}$. 
Furthermore, consider two different $M_1, M_2 \in M$. There exists an $x \in S$ which is a row of $M_1$ but not a row of $M_2$.
Thus, $\|M_1 x\|_p \ge d$ and
\[
\|M_2 x\|_p \le C_p\sqrt{d} \cdot R^{1 / p} < d / C. \qedhere
\]
\end{proof}

\begin{theorem}\label{thm:selb}
Solving the $\ell_p$ subspace sketch problem requires $\widetilde{\Omega}(d^{p / 2})$ bits when $0 < \eps < 1$ and $p \ge 2$ are constants and $n=\Omega(d^{p/2})$.
\end{theorem}

\begin{proof}
We first prove a lower bound for randomized data structures for the $\ell_p$ subspace sketch problem with failure probability $d^{-p} / 100$. Let $\mathcal{M}\subset \R^{R\times d}$ and $S\subset \bcube^d$ be as constructed in Lemma~\ref{lem:selb_hard_set}.
Choose a matrix $M$ from $\mathcal{M}$ uniformly at random. 
Since for each $x \in  \bcube^d$, with probability at least $1 - d^{-p}  / 100$, 
\begin{equation}\label{equ:approx}
\|Mx\|_p^p \le Q_p(x) \le (1+ O(\eps)) \|Mx\|_p^p,
\end{equation}
by a union bound, with probability at least $0.99$, \eqref{equ:approx} holds simultaneously for all $x \in S$. It follows from Lemma~\ref{lem:selb_hard_set}(i) that by querying $\|Mx\|_p$ for all $x\in S$, one can distinguish all different $M \in \mathcal{M}$. A standard information-theoretic argument leads to a lower bound of $\Omega(\log |\mathcal{M}|) = \Omega(d^{p / 2})$.

For randomized data structures for the $\ell_p$ subspace sketch problem with constant failure probability, a standard repetition argument implies that the failure probability can be reduced to $d^{-p} / 100$ using $O(\log d)$ independent repetitions. Therefore a lower bound of $\widetilde{\Omega}(d^{p / 2})$ bits follows.
\end{proof}

\begin{remark}\label{rem:d^{p/2}_foreach_tight}
The lower bound in Theorem \ref{thm:selb} is nearly optimal.
To obtain an $\ell_p$ subspace sketch with constant $\eps$ and $\widetilde{O}(d^{p / 2})$ bits, one can first apply Lewis weights sampling \cite{lewis_sampling} to reduce the size of $A$ to $\widetilde{O}(d^{p / 2}) \times d$, and then apply the embedding in  \cite{gw18} to further reduce the number of rows of $A$ to $\widetilde{O}(d^{(p / 2) \cdot (1 - 2 / p)}) = \widetilde{O}(d^{p / 2 - 1} )$. 
Therefore the data structure takes $\widetilde{O}(d^{p/2})$ bits to store.
\end{remark}

\subsection{Lower Bounds for the For-All Version} \label{sec:forall_lb_poly(d)}
In this section, we prove a lower bound on the for-all version of the $\ell_p$ subspace sketch problem for the case of $p \geq 2$ and constant $\eps$. 
In the for-all version of the $\ell_p$ subspace sketch problem, the data structure $Q_p$ is required to, with probability at least $0.9$, satisfy $Q_p(x) = (1\pm \eps)\|Ax\|_p$ simultaneously for all $x\in\R^d$. 
\subsubsection{Lower Bound for $p \geq 2$}\label{sec:p>=2}
Throughout this section we assume that $p\geq 2$ is a constant.

Let $N = c_pd^{p/2}$ in this section, where $c_p > 0$ is a constant that depends only on $p$. Denote the unit ball in $\ell_p^n$ by $B_p^n$. For each $x\in B_2^n$, we define a function $f_x : \R^{N \times d} \to \R$ by 
\[
f_x(A) = \|Ax\|_p.
\]

\begin{lemma}\label{lem:f_property}
The function $f_x(\cdot)$ satisfies the following properties:
\begin{enumerate}[topsep=0.5ex,itemsep=-0.5ex,partopsep=1ex,parsep=1ex,label=(\roman*)]
\item $\E[f_x(A)] \leq Cc_p^{1/p}\sqrt{p}\sqrt{d}$, where entries of $A$ are \iid Rademacher random variables and $C$ is a constant that depends only on $c_p$;
\item $f_x(\cdot)$ is $1$-Lipschitz with respect to the Frobenius norm;
\item $f_x(\cdot)$ is a convex function.
\end{enumerate}
\end{lemma}
\begin{proof}
By Khintchine's inequality, $(\E|(Ax)_i|^p)^{1/p} \leq C\sqrt{p} \|x\|_2 = C \sqrt{p}$, where $C$ is an absolute constant. It follows that $\E\|Ax\|_p^p \leq N (C\sqrt{p})^p$ and by Jensen's inequality, $\E\|Ax\|_p \leq (\E\|Ax\|_p^p)^{1/p} \leq N^{1/p} C\sqrt{p} = Cc_p^{1/p}\sqrt{p} \sqrt{d}$, which implies (i).
To prove (ii), note that 
$$
f_x(A - B) = \|Ax - Bx\|_p \le \|Ax - Bx\|_2 \le \|A - B\|_2 \le \|A - B\|_F.
$$
(iii) is a simple consequence of the convexity of the $\ell_p$ norm. 
\end{proof}

The following lemma is a direct application of Talagrand's concentration inequality (Lemma~\ref{lem:talagrand}) with Lemma~\ref{lem:f_property}.\begin{lemma}\label{lem:f_concentrate} Let $A \in \R^{N \times d}$ and $x \in \R^d$ have \iid Rademacher random variables. It holds that
\[
\Pr_{A, x}\left\{f_x(A) \geq C c_p^{1/p} \sqrt{p} d\right\} \le e^{-cd},
\]
where $C$ is an absolute constant and $c_p$ is a constant depending only on $p$.
\end{lemma}

\begin{proof}
Let $\hat{x} = x / \sqrt{d}$. We have $\|\hat{x}\|_2 = 1$.
By Lemma \ref{lem:f_property} and Lemma \ref{lem:talagrand}, we have
\[
\Pr_{A}\left\{f_{\hat{x}}(A) \geq C c_p^{1/p} \sqrt{p} \sqrt{d}\right\} \le e^{-cd}.
\]
Since $f_x(A) = \sqrt{d} f_{\hat{x}}(A)$, we have
\[
\Pr_{A}\left\{f_x(A) \geq C c_p^{1/p} \sqrt{p} d \right\} \le e^{-cd},
\]
which implies the stated lemma.
\end{proof}
\begin{lemma}\label{lem:pm}
There exists a multiset $\mathcal{S} \subseteq \{+1, -1\}^{N \times d}$ such that
\begin{enumerate}[topsep=0.5ex,itemsep=-0.5ex,partopsep=1ex,parsep=1ex,label=(\roman*)]
\item $|\mathcal{S}| \geq \exp(c_1Nd)$;
\item For any $S, T \in \mathcal{S}$ such that $S\neq T$, there exists $i \in [N]$, such that $\|ST_i\|_p \leq C c_p^{1/p} \sqrt{p} d$;
\item When $p > 2$, for any $S\in \mathcal{S}$, $\kappa(S)\leq 2$.
\end{enumerate}
\end{lemma}
\newcommand{\bad}{\mathsf{Bad}}
\begin{proof}
We first define a set of bad matrices $\bad \subseteq \{+1, -1\}^{N \times d}$ to be
\[
\bad = \left\{A \in  \{+1, -1\}^{N\times d}: \Pr_x\left\{\|Ax\|_p \ge C c_p^{1/p} \sqrt{p} d\right\} \ge 3 e^{-cd}\right\},
\]
where $x \in \{+1, -1\}^d$ is an \iid Rademacher vector and $C_p, c$ are the same constants in Lemma~\ref{lem:f_concentrate}.
It follows from Lemma \ref{lem:f_concentrate} that
\[
\Pr_A \{A \in \bad\} \le \frac{1}{3},
\]
since otherwise
\[
\Pr_{A, x}\left\{f_x(A) \geq C c_p^{1/p} \sqrt{p} d\right\}  \ge \Pr_A \{A \in \bad\} \cdot \Pr_{x}\left\{f_x(A) \geq C c_p^{1/p} \sqrt{p} d \mid A \in \bad \right\} > e^{-cd}.
\]
Let the multiset $\mathcal{T} \subseteq \{+1, -1\}^{N \times d}$ of size $|\mathcal{T}| = \exp(c_2 Nd)$ consist of independent uniform samples of matrices in $\{+1,-1\}^{N\times d}$. We define three events as follows.
\begin{itemize}[itemsep=0pt]
	\item $\mathcal{E}_1$: $|\mathcal{T} \setminus \bad|\geq |\mathcal{T}|/3$;
	\item $\mathcal{E}_2$: For each $S \in \mathcal{T} \setminus \bad$ and each $T \in \mathcal{T} \setminus \{S\}$, there exists some $i \in [N]$ such that $\|ST_i\|_p \leq Cpd$.
	\item $\mathcal{E}_3$: There are at least $(5/6)|\mathcal{T}|$ matrices $T\in \mathcal{T}$ such that $\kappa(T)\leq 2$.
\end{itemize}
We analyze the probability of each event below.

First, notice that $\E|\mathcal{T} \cap \bad| \leq |\mathcal{T}|/3$. Thus, by Markov's inequality we have $\Pr(|\mathcal{T} \cap \bad| \geq 2|\mathcal{T}| / 3) \le 1/2$, which implies $\Pr(\mathcal{E}_1^c) \leq 1/2$.

Next, consider a fixed matrix $S \in \{+1, -1\}^{N \times d} \setminus \bad$. For a random matrix $T  \in \{+1, -1\}^{N \times d}$ whose entries are \iid Rademacher random variables, for each row $T_i$ of $T$, by definition of $\bad$, we have
\[
\Pr\left\{\|ST_i\|_p \ge C c_p^{1/p} \sqrt{p} d\right\} \le 3e^{-cd}.
\]
Since the rows of $T$ are independent, 
\[
\Pr\left\{\|ST_i\|_p \ge C c_p^{1/p} \sqrt{p} d, ~ \forall i \in [N]\right\} \le 3^N e^{-cNd} \leq e^{-c'Nd}.
\]
Choosing appropriate constants for $C$ and $c$ (and thus $c'$) allows for a union bound over all pairs $S \in \mathcal{T} \setminus \bad$ and $T \in \mathcal{T} \setminus \{S\}$, and we have $\Pr(\mathcal{E}_2^c) \le 1/3$.

Last, for the condition number, recall the classical result that for a random matrix $T$ of \iid Rademacher entries, it holds with probability $\geq 1-\exp(-c_3d)$ that $s_{\min}(T)\geq \sqrt{N}-c_4\sqrt{d}$ and $s_{\min}(T)\geq \sqrt{N}+c_4\sqrt{d}$, which implies that $\kappa(T)\leq (\sqrt{N}+c_4\sqrt{d})/(\sqrt{N}-c_4\sqrt{d})\leq 2$ when $d$ is sufficiently large. Letting $\mathcal{T}_1 = \{T\in |\mathcal{T}|: \kappa(T) > 2\}$, we have $\E|\mathcal{T}_1| \leq e^{-c_3d}|\mathcal{T}| $. Thus by a Markov bound, $\Pr\{|\mathcal{T}_1|\geq 6e^{-c_3d}|\mathcal{T}|\}\leq 1/10$, and thus $\Pr(\mathcal{E}_3^c)\leq 1/10$.

Since $\Pr(\mathcal{E}_1^c)+\Pr(\mathcal{E}_2^c)+\Pr(\mathcal{E}_3^c) < 1$, there exists a set $\mathcal{T}$ for which all $\mathcal{E}_1$, $\mathcal{E}_2$, $\mathcal{E}_3$ hold. Taking $\mathcal{S}$ to be the well-conditioned matrices in $\mathcal{T} \setminus \bad$, we see that $\mathcal{S}$ satisfies conditions (i)--(iii).
\end{proof}

\begin{theorem}\label{thm:comm_lb_p>=2}
The for-all version of the $\ell_p$ subspace sketch problem requires $\Omega((d / p)^{p / 2} \cdot d)$ bits to solve when $p \ge 2$ and $\eps < 1$ are constants and $n=\Omega(d^{p/2})$.
The lower bound holds even when $\kappa(A) \le 2$ if $p > 2$, and all entries in $A$ are in $\{+1, -1\}$.
\end{theorem}
\begin{proof}
Choose a matrix $A$ uniformly at random from the set $\mathcal{S}$ in Lemma~\ref{lem:pm}. Suppose that $Q_p : \R^d \to \R$ satisfies 
$$
\|Ax\|_p^p \le Q_p(x) \le (1 + O(\eps)) \|Ax\|_p^p, \quad x\in\R^d.
$$
For any row $A_i \in \R^d$ of $A$, interpreted as a column vector, $\|A A_i\|_p \ge d$, whereas for any $B \in \mathcal{S} \setminus \{A\}$, there exists a row $A_i$ of $A$ such that $\|BA_i\|_p \le C c_p^{1/p} \sqrt{p} d < d / 3$, provided that $c_p$ (and thus $C$) is small enough. Thus, by appropriate choice of the constants in Lemma~\ref{lem:pm}, we can use $Q_p$ to determine which matrix $A \in \mathcal{S}$ has been chosen. By Property (ii) of the set $\mathcal{S}$, it must hold that all elements of $\mathcal{S}$ are distinct from each other. It then follows from a standard information-theoretic argument that the size of the data structure for the $\ell_p$ sketch problem is 
lower bounded by $\Omega(\log |\mathcal{S}|) = \Omega(Nd) = \Omega((d / p)^{p / 2} \cdot d)$.
\end{proof}

\subsubsection{Lower Bound for $1\leq p \le 2$}\label{sec:p<2}
The lower bound for $1\leq p < 2$ follows from the lower bound for $p=2$ by embedding $\ell_p$ into $\ell_2$. It is known that $\ell_2^n$ $K$-embeds into $\ell_p^m$ for some $m \leq cn$, where $c = c(p)$ and $K = K(p)$ are constants that depend only on $p$. Furthermore, the embedding $T:\ell_2^n\to \ell_p^m$ can be realized using a rescaled matrix of \iid Rademacher entries (with high probability). See~\cite[Section 2.5]{matousek} for a proof for $p=1$, which can be generalized easily to a general $p$. 
Thus, one can reduce the for-all version of the $\ell_2$ subspace sketch problem to the for-all version of the $\ell_p$ subspace sketch with $1 \le p \le 2$.
Thus the lower bound of $\Omega(d^2)$ also holds when $1 \le p \le 2$.

% !TEX root = main.tex
\section{Linear Embeddings}
In this section, our goal is to show that isomorphic embeddings into low-dimensional spaces induce solutions to the subspace sketch problem. Therefore a lower bound on the subspace sketch problem implies a lower bound on the embedding dimension.

\begin{theorem}\label{thm:conversion_to_dimension_lb}
Let $p,q\geq 1$, $\eps > 0$ and $A\in \R^{N\times d}$ with full column rank. Let $E\subseteq \ell_p^N$ be the column space of $A$ and suppose that $T: E\to \ell_q^n$ is a $(1+\eps )$-isomorphic embedding. Then there exists a data structure for the for-all version of the $\ell_p$ subspace sketch problem on $A$ with approximation ratio $1 \pm 6 p \varepsilon$ and $O(nd\log(N^{|1/p-1/2|}dn\kappa(A)/\eps))$ bits. 
\end{theorem}
\begin{proof}
Without loss of generality, we may assume that $\frac{1}{\kappa(A)}\|x\|_2\leq \|Ax\|_2\leq \|x\|_2$. Let $B\in \R^{n\times d}$ be such that $B = TA$. Then $\|Ax\|_p^p \leq \|Bx\|_q^p \leq (1+\eps)^p\|Ax\|_p^p$. Round each entry of $B$ to an integer multiple of $\delta = \eps/(D_1 n^{1/q} d^{1/2} \kappa(A))$, where $D_1 = \max\{1, N^{1/2-1/p}\}$, obtaining $\tilde B$. First we claim that the rounding causes only a minor loss, 
\begin{equation}\label{eqn:tilde B preserves}
(1-\eps)^p\|Bx\|_q^p \leq \|\tilde Bx\|_q ^p \leq (1+\eps)^p\|Bx\|_q^p.
\end{equation}
Indeed, write $B = \tilde B + \Delta B$, where each entry of $\Delta B$ is bounded by $\delta$. Then
\[
\|(\Delta B)x\|_q \leq
n^{1/q}d^{1/2} \delta \|x\|_2 \leq n^{1/q} d^{1/2} \delta \|x\|_2 \leq n^{1/q} d^{1/2} \delta\cdot \kappa(A)\cdot \|Ax\|_2 \leq \eps \|Ax\|_p \le \eps \|Bx\|_q.
\]
This proves \eqref{eqn:tilde B preserves} and so
\[
(1- \eps)^p\|Ax\|_p^p \leq \|\tilde Bx\|_q^p \leq (1+\eps)^{2p}\|Ax\|_p^p,
\]
which implies that the matrix $\tilde B$ can be used to solve the $\ell_p$ subspace sketch problem on $A$ with approximation ratio $1 \pm 6p\eps$. Since 
\[
\|Bx\|_q\leq (1+\eps)^p\|Ax\|_p \leq (1+\eps)^pD_2\|Ax\|_2 \leq (1+\eps)^p D_2 \|x\|_2,
\]
where $D_2=\max\{1,N^{1/p-1/2}\}$, each entry of $B$ is at most $eD_2$. Hence, after rounding, each entry of $\tilde B$ can be described in $O(\log(D_2/\delta)) = O(\log(dnD\kappa(A)/\epsilon))$ bits, where $D = D_1D_2 = N^{|1/2-1/p|}$. The matrix $\tilde B$ can be described in $O(nd\log(D_2/\delta))$ bits. The value of $\delta$ can be described in $O(\log(1/\delta))$ bits, which is dominated by the complexity for describing $\tilde B$.
Therefore the size of the data structure is at most $O(nd\log(D_2/\delta)) = O(nd\log(Ddn\kappa(A)/\epsilon))$ bits.
\end{proof}

The dimension lower bound for linear embeddings now follows as a corollary from combining the preceding theorem with Theorem~\ref{thm:comm_lb_eps}, where we choose $d = C\log(1/\eps)$ and note that $N = O(1/\epsilon^2)$ and $\kappa(A) = O(1)$ in our hard instance.
\begin{corollary}\label{cor:dim_lb} Let $p\in [1,\infty)\setminus 2\Z$ and suppose that $d\geq C\log(1/\eps)$. It holds that
\[
	N_p(d,\eps)\geq c_p \cdot 1/(\eps^2 \cdot \polylog(1/\eps)),
\]
where $c_p > 0$ is a constant that depends only on $p$ and $C>0$ is an absolute constant.
\end{corollary}

\begin{remark}\label{rem:N_1_bounds}
It is not clear how much the assumption $d\geq C\log\frac{1}{\eps}$ can be weakened. The best known results for $p=1$ are as follows~\cite{handbook:21}.
\[
N_1(d,\eps)\leq \begin{cases}
									c_2 \eps^{-1/2}, & \quad d=2;\\
									c(d)\left(\frac{1}{\eps^2}\log\frac{1}{\eps}\right)^{(d-1)/(d+2)}, & \quad d=3,4;\\
									c(d)\left(\frac{1}{\eps^2}\right)^{(d-1)/(d+2)}, &\quad d\geq 5,
									\end{cases}
\]
which is substantially better than $1/(\eps^2\polylog(1/\eps))$ for constant $d$. In a similar lower bound~\cite{bourgain:zonoid_lb}
\[
N_1(d,\eps)\geq c(d)\eps^{-2(d-1)/(d+2)},
\]
the constant $c(d)\approx e^{-c'd\ln d}$, so the lower bound is nontrivial for $d$ up to $O\left(\log\frac{1}{\eps}/\log\log\frac{1}{\eps}\right)$. Since $N_1(d,\eps)$ is increasing, optimizing $d$ w.r.t.~$\eps$ yields that $N_1(d,\eps) = \Omega\big(\eps^{-2}\exp(-c''\sqrt{\ln(1/\eps)\ln\ln(1/\eps)})\big)$ for all $d = \Omega\big(\sqrt{\ln(1/\eps)}\big)$.  Our result improves the lower bound to $\eps^{-2}/\polylog(1/\epsilon)$ for larger $d$ and, more importantly, works for general $p\geq 1$ that is not an even integer.
\end{remark}

\begin{remark}
In the case of $p > 2$, it is an immediate corollary from Theorem~\ref{thm:comm_lb_p>=2} that $N_p(d,\epsilon) = \Omega(d^{p/2}/\log d)$, which recovers the known (and nearly tight) lower bound up to a logarithmic factor.
\end{remark}

%!TEX root = main.tex
\section{Sampling-based Embeddings}

Our goal in this section is to prove the following lower bound.

\begin{theorem}\label{thm:sampling_based}
Let $p\geq 1$ and $p\notin 2\Z$. Suppose that $Q_p(x) = \|TAx\|_p^p$ solves the for-all version of the $\ell_p$ subspace sketch problem on $A$ and $\eps$ for some $T\in \R^{m\times n}$ such that each row of $T$ contains exactly one non-zero element. Then it must hold that $m\geq c_p d/(\epsilon^2\polylog(1/\epsilon))$, provided that $d\geq C \log(1/\epsilon)$, where $c_p > 0$ is a constant that depends only on $p$ and $C > 0$ is an absolute constant.
\end{theorem}
\begin{proof}
Let $A$ be the hard instance matrix for the $\widetilde\Omega(1/\eps^2)$ lower bound in Corollary~\ref{cor:dim_lb}. Recall that $A$ is a diagonal matrix with $k = \Theta(d/\log(1/\eps))$ diagonal blocks. Each block has dimension $2^s\times s$ with $2^{s/2} = 1/\eps^{1-o(1)}$. Furthermore, each block can be written as $DB$, where $D$ is a $2^s\times 2^s$ diagonal matrix and $B$ is the $2^s\times s$ matrix whose rows are all vectors in $\{-1,1\}^s$, and each entry in $D$ has magnitude in $[(2\sqrt{s})^{1/p},(8\sqrt{s})^{1/p}]$. The matrix $B$ can be described using $O(s2^s)$ bits and the matrix $D$ using $O(2^s\log s)$ bits. Without loss of generality we may assume that each nonzero entry of $T$ is an integer multiple of $\eps^{1/p}$, since the loss of rounding, by the triangle inequality, is at most $\eps\|Ax\|_p^p$. Next, we shall bound the number of bits needed to describe $T$.

Letting $x=e_j$ be a canonical basis vector,
\[
(1+\epsilon)\|Ax\|_p^p \geq \|TAx\|_p^p = \sum_{i=1}^m |t_i A_{i,j}|^p\geq 2\sqrt{s} \sum_{i=1}^m |t_i|^p.
\]
On the other hand,
\[
\|Ax\|_p^p \leq k\cdot 8\sqrt{s} \|Bx\|_p^p\leq 8 k \sqrt{s}(2^{s/2}\|Bx\|_2)^p \leq C'k \sqrt{s}2^{sp},
\]
where we used \eqref{eqn:khintchine_B} for the last inequality. It follows from the AM-GM inequality that
\[
\sum_i \log\frac{|t_i|}{\epsilon^{1/p}} = \log \prod_i\frac{|t_i|}{\epsilon^{1/p}} \leq \frac{m}{p}\log\left(\frac{1}{m} \sum_i \frac{|t_i|^p}{\eps} \right) \leq C'' m\left(s + \log\frac{1}{\eps} + \log\frac{k}{m}\right)\leq C'''ms,
\]
that is, $T$ can be described using $O(ms)$ bits, provided that $m\geq k$. Therefore $TA$ can be described in $O(s2^s + 2^s\log s + ms) = O((m+1/\eps^2)\log(1/\epsilon))$ bits. Combining with the lower bound of $\Omega(d/(\eps^2\poly\log(1/\eps))$ bits, we see that $m=\Omega(d/(\eps^2\polylog(1/\eps)))$.

A similar argument shows that when $m < k$, the matrix $T$ can be described in $O(d)$ bits, which leads to a contradiction to the lower bound. Hence it must hold that $m\geq k$ and, as we proved above, $m=\Omega(d/(\eps^2\polylog(1/\eps)))$.
\end{proof}

The lower bound for the (for-each) $\ell_p$ subspace sketch problem loses further a factor of $\log d$.
\begin{corollary} 
Let $p\geq 1$ and $p\notin 2\Z$. Suppose that $Q_p(x) = \|TAx\|_p^p$ solves the (for-each) version of the $\ell_p$ subspace sketch problem on $A$ and $\eps$ for some $T\in \R^{m\times n}$ such that each row of $T$ contains exactly one non-zero element. Then it must hold that $m\geq c_p d/(\epsilon^2\cdot\log d\cdot\polylog(1/\epsilon))$, provided that $d\geq C\log(1/\epsilon)$, where $c_p > 0$ is a constant that depends only on $p$ and $C > 0$ is an absolute constant.
\end{corollary}
\begin{proof}
Observe that we used the approximation to $\|Ae_i\|_p^p$ for each canonical basis vector $e_i$ in the proof of Theorem~\ref{thm:sampling_based}, which holds with a constant probability if we make $O(\log d)$ independent copies of the (randomized) data structure. This incurs a further loss of a $\log d$ factor in the lower bound.
\end{proof}

%!TeX root = main.tex

\section{Oblivious Sketches}\label{sec:OSE}

An oblivious subspace embedding for $d$-dimensional subspaces $E$ in $\ell_p^n$ is a distribution on linear maps $T:\ell_p^n\to\ell_p^m$ such that it holds for any $d$-dimensional subspace $E\subseteq \ell_p^n$ that
\[
\Pr_T\left\{ (1-\eps)\|x\|_p \leq \|T x\|_p \leq (1+\eps)\|x\|_p,\ \forall x\in E \right\} \geq 0.99.
\]

More generally, an \emph{oblivious sketch} is a distribution on linear maps $T:\ell_p^n\to \R^m$, accompanied by a recovery algorithm $\mathcal{A}$, such that it holds for any $d$-dimensional subspace $E\subseteq \ell_p^n$ that
\[
\Pr_T\left\{ (1-\eps)\|x\|_p \leq \mathcal{A}(T x) \leq (1+\eps)\|x\|_p,\ \forall x\in E \right\} \geq 0.99.
\]
It is clear that an oblivious embedding is a special case of an oblivious sketch, where $\mathcal{A}(Tx) = \|Tx\|_p$.

In this section we shall show that when $1\leq p < 2$, any oblivious sketch requires $m = \widetilde{\Omega}(d/\eps^2)$.

Before proving the lower bound, let us prepare some concentration results. We use $\bS^{n-1}$ to denote the unit sphere in $(\R^n, \|\cdot\|_2)$. First, observe that the norm function $x\mapsto \|x\|_p$ is a Lipschitz function of Lipschitz constant $\max\{1,n^{1/p-1/2}\}$. Also note that $(\E_{g\sim N(0,I_n)} \|g\|_p^p)^{1/p} = \beta_p n^{1/p}$, where $\beta_p = (\E_{g\sim N(0,1)} |g|^p)^{1/p}$. Standard Gaussian concentration (Lemma~\ref{lem:gaussian_lipschitz}) leads to the following:
\begin{lemma}
Let $p\geq 1$ be a constant and $g\sim N(0,I_n)$. It holds with probability at least $1-\exp(-c\epsilon^2 n^{\min\{1,2/p\}})$ that $(1-\eps) \beta_p n^{1/p} \leq \|g\|_p\leq (1+\eps) \beta_p n^{1/p}$, where $c=c(p)>0$ is a constant that depends only on $p$.
\end{lemma}

Suppose that $G$ is an $n\times d$ Gaussian random matrix of \iid $N(0,1)$ entries. Observe that for a fixed $x\in \bS^{d-1}$, $Gx\sim N(0,I_n)$. A typical $\eps$-net argument on $\bS^{d-1}$ allows us to conclude the following lemma. We remark that this gives Dvoretzky's Theorem for $\ell_p$ spaces.
\begin{lemma}
Let $1\leq p < 2$ be a constant and $G$ be an $n\times d$ Gaussian random matrix. There exist constants $C = C(p) > 0$ and $c = c(p) > 0$ such that whenever $n\geq Cd\log(1/\eps)/\eps^2$, it holds $\Pr\left\{ (1-\eps)\beta_p n^{1/p} \leq \|Gx\|_p \leq (1+\eps)\beta_p n^{1/p},\ \forall x\in \bS^{d-1} \right\}\geq 1 - 2\exp(-c \eps^2 n)$.
\end{lemma}

Now, consider two distributions on $n\times d$ matrices, where $n = \Theta(d\epsilon^{-2}\log(1/\eps))$. The first distribution $\cL_1$ is just the distribution of a Gaussian random matrix $G$ of \iid $N(0,1)$ entries, and the second distribution $\cL_2$ is the distribution of $G+\sigma uv^T$, where $G$ is the Gaussian random matrix of \iid $N(0,1)$ entries, $u\sim N(0,I_n)$ and $v\sim N(0,I_d)$ and $\sigma = \alpha\sqrt{\eps/d}$ for some constant $\alpha$ to be determined later, and $G$, $u$ and $v$ are independent. 

\begin{theorem}\label{thm:OSE_lb}
 Let $1\leq p < 2$ be a constant. Suppose that $S\in \R^{m\times n}$ is an oblivious sketch for $d$-dimensional subspaces in $\ell_p^n$, where $n=\Theta(d\eps^{-2}\log(1/\eps))$. It must hold that $m \geq c d/\eps^2$, where $c=c(p)>0$ is a constant depending only on $p$.
\end{theorem}
\begin{proof}
It follows from the preceding lemma that, if $A\sim \cL_1$, we have $\sup_{x\in \bS^{d-1}} \|Ax\|_p \leq (1+\eps)\beta_p n^{1/p}$ with probability at least $0.999$ with an appropriate choice of constant in the $\Theta$-notation of $n$. Next we consider the supremum of $\|Ax\|_p$ when $A\sim \cL_2$. Observe that
\[
\sup_{x\in \bS^{d-1}} \left\|(G+\sigma uv^T)x\right\|_p\geq \left\|(G+\sigma uv^T)\frac{v}{\|v\|_2}\right\|_p = \left\|G\frac{v}{\|v\|_2} + \sigma u \|v\|_2 \right\|_p.
\]
Since $v\sim N(0,I_d)$, the direction $v/\|v\|_2\sim \Unif(\bS^{d-1})$ and the magnitude $\|v\|_2$ are independent, and by rotational invariance of the Gaussian distribution, $Gx\sim N(0,I_d)$ for any $x\in \bS^{d-1}$. Hence
\[
\left\|G\frac{v}{\|v\|_2} + \sigma u \|v\|_2 \right\|_p \eqdist \|u_1 + \sigma t u_2\|_p \eqdist \sqrt{1+\sigma^2 t^2}\|u\|_p,
\]
where $t$ follows the distribution of $\|v\|_2$ and $u_1,u_2$ are independent $N(0,I_n)$ vectors. Applying the preceding two lemmata, we see that with probability at least $0.998$, it holds that $t\geq 0.99\sqrt{d}$ and $\|u\|_p\geq (1-\eps)\beta_p n^{1/p}$. Therefore, when $A\sim \cL_2$, with probability at least $0.998$, we have $\sup_{x\in \bS^{d-1}}\|Ax\|_p\geq \sqrt{1+0.99^2\alpha^2\eps}(1-\eps)\beta_p n^{1/p}\geq (1+4\eps)\beta_p n^{1/p}$, for an appropriate choice of $\alpha$.

Therefore with the corresponding recovery algorithm $\mathcal{A}$,
\begin{gather*}
\Pr_{A\sim \cL_1,S}\left\{\sup_{x\in \bS^{n-1}} \mathcal{A}(SAx) \leq (1+\eps)^2\beta_p n^{1/p}\right\}\geq 0.9,\\
\Pr_{A\sim \cL_2,S}\left\{\sup_{x\in \bS^{n-1}} \mathcal{A}(SAx) \geq (1+4\eps)(1-\eps)\beta_p n^{1/p}\right\}\geq 0.9,
\end{gather*}
which implies that the linear sketch $S$ can be used to distinguish $\cL_1$ from $\cL_2$ by evaluating $\sup_{x\in \bS^{d-1}} \mathcal{A}(SAx)$. It then follows from~\cite[Theorem 4]{LW16:random} that the size of the sketch $md\geq c/\sigma^4 = c'd^2/\eps^2$ for some absolute constants $c,c'>0$, and thus $m\geq c'd/\eps^2$.
\end{proof}
%
%Let $\eta = \eps/\sqrt{\log(1/\eps)}$. Then $\eps^2 \leq \eta^2 \log(1/\eta)$. The theorem above can be rephrased as
%\begin{reptheorem}{thm:OSE_lb}[rephrased]
%Let $1\leq p < 2$ be a constant. Suppose that $S\in \R^{m\times n}$ is an oblivious sketch for $d$-dimensional subspaces in $\ell_p^n$, where $n=\Theta(d/\eps^2)$. It must hold that $m \geq c d / (\eps^2\log(1/\eps))$, where $c=c(p)>0$ is a constant depending only on $p$.
%\end{reptheorem}

%!TeX root = main.tex
\section{Lower Bounds for $M$-estimators}\label{sec:m-est}
The main theorem of this section is the following.

\begin{theorem}\label{thm:power_function}
Suppose there exist $\alpha,\lambda>0$ and $p\in (0,\infty)\setminus 2\Z$ such that $\phi(t/\lambda)\sim \alpha |t|^p$ as $t\to \infty$ or $t\to 0$. Then the subspace sketch problem for $\Phi(x) = \sum_{i = 1}^n \phi(x_i)$ requires $\Omega(d/(\eps^2\polylog(1/\eps)))$ bits when $d\geq C_1\log(1/\eps)$ and $n\geq C_2 d/(\eps^2\polylog(1/\eps))$ for some absolute constants $C_1,C_2 > 0$.
\end{theorem}

\begin{proof}
We reduce the problem to the $\ell_p$ subspace sketch problem. We prove the statement in the case of $t\to\infty$ below. The proof for the case of $t\to 0$ is similar.

For a given $\eps > 0$, there exists $M$ such that $(1-\eps)\alpha |t|^p \leq \phi(t/\lambda) \leq (1+\eps)\alpha|t|^p$ for all $|t|\geq M$. Let $A$ be our hard instance for the $\ell_p$ subspace sketch problem in Theorem~\ref{thm:comm_lb_eps}. Then each row of $A$ is a $\{-1,1\}$-vector scaled by a factor of $\tilde y_i\geq \Delta$ for some $\Delta =\Omega\big( \log^{1/(2p)}(1/\eps)\big)$. One can recover a random sign used in the construction of $A$ by querying $Ax$ for a $\{-1,1\}$-vector $x$. Therefore, if $(Ax)_i\neq 0$, it must hold that $|(Ax)_i|\geq \Delta$. This implies that there exists a scaling factor $\beta = M/\Delta$ such that $(1-\eps)\alpha\|\beta Ax\|_p^p \leq \Phi(\lambda^{-1}\beta Ax) \leq (1+\eps)\alpha\|\beta Ax\|_p^p$, that is, $\alpha^{-1}\beta^{-p}\Phi(\lambda^{-1}\beta Ax)$ is a $(1\pm \eps)$-approximation to $\|Ax\|_p^p$ for $\{-1,1\}$-vectors $x$. The conclusion follows from Theorem~\ref{cor:comm_lb_eps} (which plants independent copies of hard instance $A$ in diagonal blocks) and a rescaling of $\eps$.
\end{proof}

We have the following immediate corollary.
\begin{corollary}\label{cor:m_estimator1} The subspace sketch problem for $\Phi$ requires $\Omega(d/(\eps^2\polylog(1/\eps)))$ bits when $d\geq C_1\log(1/\eps)$ and $n\geq C_2 d/(\eps^2\polylog(1/\eps))$  for some absolute constants $C_1, C_2 > 0$ for the following functions $\phi$:
\begin{itemize}[topsep=0.5ex,itemsep=-0.5ex,partopsep=1ex,parsep=1ex]
	\item ($L_1$-$L_2$ estimator) $\phi(t) = 2(\sqrt{1+t^2/2}-1)$;
	\item (Huber estimator) $\phi(t) = t^2/(2\tau)\cdot \mathbf{1}_{\{|t|\leq \tau\}} + (|t|-\tau/2)\cdot\mathbf{1}_{\{|t| > \tau\}}$;
	\item (Fair estimator) $\phi(t) = \tau^2(|x|/\tau - \ln(1+|t|/\tau))$;
	\item (Tukey loss $p$-norm) $\phi(t) = |t|^p\cdot \mathbf{1}_{\{|t|\leq \tau\}} + \tau^p\cdot\mathbf{1}_{\{|t| > \tau\}}$.

\end{itemize}
\end{corollary}

Now we prove the $\Omega(d/(\eps^2\polylog(1/\eps)))$ lower bound for the subspace sketch problem for the Cauchy estimator $\phi(t) = (\tau^2/2)\ln(1 + (t/\tau)^2)$. First consider an auxiliary function $\phi_{\text{aux}}(t) = \ln |x|\cdot \mathbf{1}_{\{|x|\geq 1\}}$, for which we shall have also an $\Omega(d/(\eps^2\polylog(1/\eps))$ lower bound by following the approach in Section~\ref{sec:1/eps^2 lb} with some changes we highlight below. Instead of $M_{i,j}^{(d,p)} = |\langle i,j\rangle|^p$, we shall define $M_{i,j}^{(d,p)} = \phi_{\text{aux}}(\langle i,j\rangle)$, and we proceed to define $N^{(d)}$ and $\Lambda_0^{(d,p)}$ in the same manner. The following lemma is similar to Corollary~\ref{cor:singular_value_lb}, showing that this new matrix $M^{(d,p)}$ also has large singular values. The proof is postponed to Section~\ref{sec:proof_log_matrix_singular_value}.

\begin{lemma}\label{lem:log_matrix_singular_value} Suppose that $d\in 8\Z$. Then $\Lambda_0^{(d,p)} \geq c2^{d/2}/\sqrt{d}$ for some absolute constant $c>0$.
\end{lemma}

Therefore, the entire lower bound argument in Corollary~\ref{cor:comm_lb_eps_p=0} goes through. 
We can then conclude that the subspace sketch problem for $\Phi_{\text{aux}}(x) = \sum_{i = 1}^n \phi_{\text{aux}}(x_i)$ requires $\Omega(d/(\eps^2\polylog(1/\eps))$ bits. Now, for the Cauchy estimator $\phi(t) = (\tau^2/2)\ln(1 + (t/\tau)^2)$, note that $(1-\eps) \tau^2 \phi_{\text{aux}}(t) \leq \phi(\tau \cdot t)  \leq (1+\eps) \tau^2 \phi_{\text{aux}}(t)$ for all sufficiently large $t$. It follows from a similar argument to the proof of Theorem~\ref{thm:power_function} that the same lower bound continues to hold for the subspace sketch problem for the Cauchy estimator. 
\begin{corollary}\label{cor:m_estimator2}
The subspace sketch problem for $\Phi$ requires $\Omega(d/(\eps^2\polylog(1/\eps)))$ bits for the Cauchy estimator $\phi(t) = (\tau^2/2)\ln(1+(t/\tau)^2)$, 
when $d\geq C_1\log(1/\eps)$ and $n\geq C_2 d/(\eps^2\polylog(1/\eps))$ for some absolute constants $C_1,C_2 > 0$.
\end{corollary}

\subsection{Proof of Lemma~\ref{lem:log_matrix_singular_value}}\label{sec:proof_log_matrix_singular_value}
Differentiate both sides of \eqref{eqn:aux2} w.r.t\@ $p$,
\begin{align*}
-\int_0^\infty \frac{(2\sin t)^{2n}\ln t}{t^{p+1}}dt &= 2\sum_{k=1}^n (-1)^k \binom{2n}{n+k} (2k)^p\ln(2k) \cos\left(\frac{\pi p}{2}\right)\Gamma(-p) \\
&\qquad - 2\sum_{k=1}^n (-1)^k \binom{2n}{n+k} (2k)^p \frac{\pi}{2}\sin\left(\frac{\pi p}{2}\right)\Gamma(-p) \\
&\qquad - 2\sum_{k=1}^n (-1)^k \binom{2n}{n+k} (2k)^p \cos\left(\frac{\pi p}{2}\right) \Gamma'(-p).
\end{align*}
Thus using the reflection identity~\eqref{eqn:reflection_identity}, 
\begin{align*}
\frac{\Gamma(p+1)}{\pi}\sin\left(\frac{\pi p}{2}\right)\int_0^\infty \frac{(2\sin t)^{2n}\ln t}{t^{p+1}}dt &= \sum_{k=1}^n (-1)^k \binom{2n}{n+k} (2k)^p\ln(2k)  \\
&\qquad - 2\sum_{k=1}^n (-1)^k \binom{2n}{n+k} (2k)^p \frac{\pi}{2}\frac{\sin^2\left(\frac{\pi p}{2}\right)}{\sin(p\pi)} \\
&\qquad - \sum_{k=1}^n (-1)^k \binom{2n}{n+k} (2k)^p \frac{\Gamma'(-p)}{\Gamma(-p)}.
\end{align*}
Letting $p\to 0^+$, we see that the middle term on the right-hand side vanishes, which implies
\begin{align*}
\frac{2^{2n}}{\pi} \lim_{p\to 0^+} \sin\left(\frac{\pi p}{2}\right)\int_0^\infty \frac{(\sin t)^{2n}\ln t}{t^{p+1}}dt &= \sum_{k=1}^n (-1)^k \binom{2n}{n+k} \ln k  \\
&\qquad -  \lim_{p\to 0^+}  \sum_{k=1}^n (-1)^k \binom{2n}{n+k} \left(\frac{\Gamma'(-p)}{\Gamma(-p)} - \ln 2\right).
\end{align*}
Invoking Lemma~\ref{lem:critical_identity}, we obtain that
\begin{equation}\label{eqn:log_aux}
\sum_{k=1}^n (-1)^{k} \binom{2n}{n+k} \ln k 
= \frac{2^{2n}}{\pi} \lim_{p\to 0^+} \sin\left(\frac{\pi p}{2}\right)\int_0^\infty \left(\frac{(\sin^{2n} t)\ln t}{t^{p+1}} - \left(\frac{\Gamma'(-p)}{\Gamma(-p)} -\ln 2\right) \frac{\sin^{2n}t}{t^{p+1}} \right) dt.
\end{equation}
We claim that the limit on the right-hand side at least $c/\sqrt{n}$ for some absolute constant $c > 0$. Note that letting $p\to 0^+$ in Lemma~\ref{lem:critical_identity} leads to
\[
\sum_{k=1}^n (-1)^{k+1} \binom{2n}{n+k} = \frac{2^{2n}}{\pi}\lim_{p\to 0^+} \sin\left(\frac{\pi p}{2}\right) \int_0^\pi \frac{\sin^{2n} t}{t^{p+1}} dt,
\]
and the left-hand side is
\[
\sum_{k=1}^n (-1)^{k+1} \binom{2n}{n+k} = \frac{1}{2}\binom{2n}{n}\sim \frac{2}{3}\cdot \frac{2^{2n}}{\sqrt n},
\]
thus
\[
\lim_{p\to 0^+} p \int_0^\pi \frac{\sin^{2n} t}{t^{p+1}} dt \sim \frac{4}{3\sqrt n}.
\]
Note the fact that $\Gamma'(x)/\Gamma(x)= -1/x - \gamma + o(1)$ as $x\to 0^+$ (e.g., plugging $n=1$ into Eq. (1.2.15) in \cite[p13]{special_functions}), where $\gamma = 0.577\cdots$ is the Euler gamma constant. Thus the limit on the right-hand side of \eqref{eqn:log_aux} is the same as
\[
\lim_{p\to 0^+} \sin\left(\frac{\pi p}{2}\right) \int_0^\infty \left(\frac{(\sin^{2n} t)\ln t}{t^{p+1}} - \left(\frac{1}{p}-\gamma-\ln 2\right) \frac{\sin^{2n}t}{t^{p+1}} \right) dt.
\]
and it suffices to show that
\begin{equation}\label{eqn:log_aux_main_limit}
\lim_{p\to 0^+} p \int_0^\infty \left(\frac{(\sin^{2n} t)\ln t}{t^{p+1}} - \frac{1}{p} \frac{\sin^{2n}t}{t^{p+1}} \right) dt > -\frac{c}{\sqrt n}
\end{equation}
for some constant $c \in \big(0,\frac{4}{3}(\gamma+\ln 2)\big)$. Since the limit above must exist, we can pick a sequence $p_k\to 0$ such that $e^{1/p_k}$ is a multiple of $\pi$. Hence we assume that $N = e^{1/p}/\pi$ is an integer below.

Now, split the integral into $[0,\pi]$ and $[\pi,\infty)$. Observe that 
\[
\lim_{p\to 0^+} \sin\left(\frac{\pi p}{2}\right) \int_0^\pi \frac{(\sin^{2n} t)\ln t}{t^{p+1}}dt = 0,\quad \lim_{p\to 0^+} \sin\left(\frac{\pi p}{2}\right) \int_0^\pi \frac{\sin^{2n} t}{t^{p+1}}dt = 0
\]
because the integrands, viewed as functions of $(p,t)$, are bounded on $[0,1]\times [0,\pi]$, since $\sin^{2n}t\sim t^{2n}$ near $t=0$ and so $t=0$ is not a singularity. Furthermore,
\[
\lim_{p\to 0^+} \sin\left(\frac{\pi p}{2}\right) \int_0^\pi \frac{1}{p}\cdot \frac{\sin^{2n} t}{t^{p+1}}dt = \frac{\pi}{2} \int_0^\pi \frac{\sin^{2n}t}{t}dt
\]
because the integrand is uniformly continuous on $[0,1]\times [0,\pi]$ and we can take the limit under the integral sign.

Now we deal with the integral on $[\pi,\infty)$. Following our approach in Corollary~\ref{cor:singular_value_lb}, we have that
\begin{align}
\int_\pi^\infty (p\ln t-1)\frac{(\sin^{2n} t)}{t^{p+1}} dt &= -\int_\pi^{N\pi} (1-p\ln t)\frac{(\sin^{2n} t)}{t^{p+1}} dt + \int_{N\pi}^\infty (p\ln t-1)\frac{(\sin^{2n} t)}{t^{p+1}} dt \notag\\
&\geq \left(-\sum_{k=1}^{N-1} \frac{1-p\ln(k\pi)}{(\pi k)^{p+1}} + \sum_{k=N+1}^\infty \frac{p\ln((k-1)\pi)-1}{(\pi k)^{p+1}} \right) I_n, \label{eqn:log_aux_2}
\end{align}
where
\[
I_n = \int_0^{\pi} \sin^{2n}t dt \sim \sqrt{\frac{\pi}{n}},
\]
and we used the fact that $(1-p\ln t)/t^{p+1}$ is nonnegative and decreasing when $\ln t\leq 1/p$.

The bracketed term on the rightmost side of \eqref{eqn:log_aux_2} is
\begin{align*}
-\sum_{k=1}^{N-1} \frac{1-p\ln(k\pi)}{(\pi k)^{p+1}} + \sum_{k=N+1}^\infty \frac{p\ln((k-1)\pi)-1}{(\pi k)^{p+1}} &= \sum_{k=1}^{N-1} \frac{p\ln(k\pi)-1}{(\pi k)^{p+1}} + \sum_{k=N+1}^\infty \frac{p\ln((k-1)\pi)-1}{(\pi k)^{p+1}}\\
&= \sum_{k=1}^\infty \frac{p\ln(k\pi)-1}{(\pi k)^{p+1}} + \sum_{k=N+1}^{\infty} \frac{p\ln(1-\frac{1}{k})}{(\pi k)^{p+1}}.
\end{align*}
The second term clearly tends to $0$ as $p\to 0^+$, while the first term is equal to
\[
\frac{-\zeta(1+p)+(p\ln\pi) \zeta(1+p)-p\zeta'(1+p)}{\pi^{1+p}} \to \frac{\ln\pi - \gamma}{\pi},\quad p\to 0^+,
\]
where $\zeta(p) = \sum_{n=1}^\infty 1/n^p$ is the Riemann zeta function and we used the fact that $\zeta(1+p) = \frac{1}{p} + \gamma + f(p)$ for an analytic function $f$ on $\mathbb{C}$ with $f(0) = 0$ (see, e.g.~\cite[p15]{special_functions}). 

Therefore we conclude that the limit on the left-hand side in~\eqref{eqn:log_aux_main_limit} is at least
\[
\frac{\ln \pi-\gamma}{\pi} I_n - \frac{\pi}{2}\int_0^\pi \frac{\sin^{2n}{t}}{t}dt 
\]
and it suffices to show that
\[
\int_0^\pi \frac{\sin^{2n}t}{t} dt \leq \frac{c_1}{\sqrt n} + o\left(\frac{1}{\sqrt n}\right)
\]
for some
\[
c_1 < \frac{2}{\pi}\left(\frac{\ln \pi-\gamma}{\sqrt{\pi}} + \frac{4}{3}(\gamma+\ln 2)\right) = 1.282\cdots.
\]
First observe that
\[
\int_{\pi/2}^\pi \frac{\sin^{2n}t}{t}dt \leq \frac{2}{\pi}\int_{\pi/2}^\pi \sin^{2n} tdt = \frac{2}{\pi}\cdot \frac{1}{2}I_n \sim \frac{1}{\sqrt{\pi n}},
\]
and
\[
\int_0^1 \frac{\sin^{2n}t}{t}dt \leq \int_0^1 t^{2n-1} dt = \frac{1}{2n}.
\]
Letting $\delta = \sqrt{(\ln n)/n}$, then $\sqrt{n} \sin^{2n}t \to 0$ as $n\to\infty$ on $t\in [1,\pi/2-\delta]$, and thus
\[
\int_1^{\frac{\pi}{2}-\delta} \frac{\sin^{2n}t}{t} dt = o\left(\frac{1}{\sqrt n}\right)
\]
and we can now bound
\[
\int_{\frac{\pi}{2}-\delta}^{\frac{\pi}{2}} \frac{\sin^{2n}t}{t} dt\leq \frac{1}{\frac{\pi}{2}-\delta}\int_{\frac{\pi}{2}-\delta}^{\frac{\pi}{2}} \sin^{2n}t dt \leq \frac{1}{\frac{\pi}{2}-\delta}\cdot \frac{1}{2}I_n\sim \frac{1}{\sqrt{\pi n}}.
\]
Therefore we conclude that we can choose $c_1$ to be
\[
c_1 = \frac{2}{\sqrt{\pi}} = 1.128\cdots
\]
as desired.

% !TEX root = main.tex

\section{Lower Bounds on Coresets for Projective Clustering}

We shall prove a lower bound of $\widetilde{\Omega}(kj/\epsilon^2)$ bits for coresets for projective clustering. First we need a lemma which provides codewords to encode the clustering information.

\begin{lemma}\label{lem:code}
For any given integer $L \ge 1$ and even integer $D \ge 2$, there exists a set $S = \{(s_1, t_1), (s_2, t_2), \ldots, (s_{m}, t_{m})\}$ of size $m \geq c 2^D/\sqrt{D}$, where $s_i, t_i \in \mathbb{R}^D$ and $c > 0$ is an absolute constant, such that 
\begin{itemize}[topsep=0.5ex,itemsep=-0.5ex,partopsep=1ex,parsep=1ex]
\item $\langle s_i, t_i\rangle = 0$;
\item $\langle s_i, t_j\rangle \ge L^2$ for $i \neq j$;
\item all entries of $s_i$ and $t_i$ are in $\{0, L\}$.
\end{itemize}
\end{lemma}
\begin{proof}
We first consider the case $L = 1$. Let $\{s_i\}$ be the set of all binary vectors with Hamming weight $D / 2$, and $t_i = \mathbf{1}^D - s_i$, i.e., $t_i$ is the complement of $s_i$.
Thus, $\langle s_i, t_i \rangle = 0$ by construction.
For any $i \neq j$, since $s_i \neq s_j$, and both $s_i$ and $s_j$ have Hamming weight $D / 2$, we have $\langle s_i, t_j \rangle \ge 1$.

For a general $L$, we replace all entries of value $1$ in the construction above with $L$.
\end{proof}

In the rest of the section, we also use an $n\times d$ matrix to represent a point set of size $n$ in $\R^d$, where each row represents a point in $\R^d$.

Below we set up the framework of the hard instance for the projective subspace clustering problem. For a given $k$, choosing $D = O(\log k)$, we can obtain a set $S$ of size $k$ as guaranteed by Lemma~\ref{lem:code}. Suppose that $j\geq D+1$ and $d\geq j+1$. Without loss of generality we may assume that $d = j+1$, otherwise we just embed our hard instance in $\R^{j+1}$ into $\R^d$ by appending zero coordinates.

For a set $\mathcal{A}$ consisting of $k$ matrices $A^{(1)}, A^{(2)}, \ldots, A^{(k)} \in \mathbb{R}^{n \times (j + 1 - D)}$, we form a point set $X = X(\mathcal{A})\in \R^{nk\times d}$, whose rows are indexed by $(i,j)\in [k]\times [n]$ and defined as
\[
X_{i,j} = \begin{pmatrix} s_i^T & A^{(i)}_j \end{pmatrix},
\]
where $A^{(i)}_j$ denotes the $j$-th row of $A^{(i)}$. 

Suppose that $y\in \R^{j+1-D}$. For each $i\in [k]$, let $V_i, W_i\subseteq \R^{j+1}$ be $j$-dimensional subspaces that satisfy
\begin{align*}
V_i \perp v_i, &\quad v_i = \begin{pmatrix}t_i & \mathbf{0}^{j+1-D}\end{pmatrix},\\
W_i \perp w_i, &\quad w_i = \begin{pmatrix}t_i & y\end{pmatrix},
\end{align*}
where, for notational simplicity, we write vertical concatenation in a row. Last, for each $\ell\in [k]$, define a center
\[
\cC_\ell = (V_1, \dots, V_{\ell-1}, W_\ell, V_{\ell+1},\dots, V_k).
\]

\begin{lemma}\label{lem:hard_instance_clustering}
When $\|y\|_2=1$ and $L^2\geq \max_i \|A^{(\ell)}_i\|_2$, it holds that $\cost(X, \cC_\ell) = \Phi(A^{(\ell)} y/ \|w_\ell\|_2)$.
\end{lemma}
\begin{proof}
One can readily verify, using Lemma~\ref{lem:code}, that $P_{ij}\perp v_i$ whenever $i\neq \ell$, and thus $P_{ij}\in V_i'$ and $\dist(P_{ij},X_\ell) = 0$ for $i\neq \ell$.

On the other hand, for $i\neq \ell$,
\[
\dist(P_{\ell j}, V_i) = \frac{|\langle P_{\ell j}, v_i\rangle|}{\|v_i\|_2} = \frac{|\langle P_{\ell j}, v_i\rangle|}{L\cdot \sqrt{D/2}} \geq \frac{L}{\sqrt{D/2}}.
\]
and
\[
\dist(P_{\ell j}, W_\ell) = \frac{|\langle P_{\ell j}, w_\ell\rangle|}{\|w_\ell\|_2} = \frac{|\langle A_i^{(\ell)}, y\rangle|}{\|w_\ell\|_2} \leq \frac{\|A^{(\ell)}_i\|_2 \|y\|_2}{\sqrt{\frac{D}{2}L^2 + \|y\|_2^2}}.
\]
Hence when $L^2\geq \|y\|_2 \max_i \|A^{(\ell)}_i\|_2$, it must hold that $W_\ell$ is the subspace in $X_\ell$ that is the closest to $P_{\ell j}$ for all $j$, and therefore
\[
\cost(X, \cC_\ell) = \sum_{j = 1}^n \phi(\dist(P_{\ell j}, W_\ell)) = \Phi\left(\frac{A^{(\ell)} y}{\|w_\ell\|_2}\right). \qedhere
\]
\end{proof}

\begin{theorem}\label{thm:clustering}
Suppose that there exists a function $\Phi$ and absolute constants $C_0$ and $\epsilon_0$ such that for any $d\geq C_0\log(k/\eps)$ and $\eps \in (0,\eps_0)$, solving the subspace sketch problem for $\Phi$ requires $M$ bits.
Then there exists an absolute constant $C_1$ such that for any $k\geq 1$ and $j \geq C_1\log(k/\eps)$, any coreset for projective clustering for $\Phi$ requires $kM$ bits.
\end{theorem}
\begin{proof}
We prove this theorem by a reduction from the subspace sketch problem for $\Phi$ to coresets for projective clustering for $\Phi$. 

Choose $D = O(\log k)$ and $d' := j + 1 - D = C_0\log(1 /\eps)$. Let $A^{(1)},\dots,A^{(k)}\in \R^{n\times d'}$ be $k$ independent hard instances for the subspace sketch problem for $\Phi$.
Let $X$ be as constructed before Lemma~\ref{lem:hard_instance_clustering}. If one can compute a projective clustering coreset for $X$ so that one can approximate $\cost(X,\cC_\ell)$ up to a $(1\pm\eps)$-factor, it follows from Lemma~\ref{lem:hard_instance_clustering} that one can approximate $\Phi(A^{(\ell)}y/\|w\|_2)$ up to a $(1\pm\eps)$-factor for every $\ell \in [k]$ and every unit vector $y\in \R^{d'}$. 
Solving the subspace sketch problem for $\Phi$ for each $A^{(\ell)}$ requires $M$ bits.
Therefore, solving $k$ independent instances requires $kM$ bits.
\end{proof}

We have the following immediate corollary.
\begin{corollary}\label{cor:clustering2}
Under the assumptions of Theorem~\ref{thm:clustering}, any coreset for projective clustering requires $\Omega(jM/\log(k/\eps))$ bits.
\end{corollary}
\begin{proof}
Let $b = j/(C_0\log(k/\epsilon))$. Let $X'$ be a block diagonal matrix of $b$ blocks, each diagonal block is an independent copy of the hard instance $X$ in Theorem~\ref{thm:clustering}. It then follows from Theorem~\ref{thm:clustering} that the lower bound is $\Omega(b M)$ bits.
\end{proof}

A lower bound of $\Omega(j k/(\eps^2 \log k \cdot \polylog(1/\eps))$ follows immediately for $\Phi(x) = \|x\|_p^p$ (Theorem~\ref{thm:comm_lb_eps}) for $p\in [0,+\infty)\setminus 2\Z^+$, and the $M$-estimators in Corollary~\ref{cor:m_estimator1} and Corollary~\ref{cor:m_estimator2}.

% !TEX root = main.tex
\section{Upper Bounds for the Tukey Loss $p$-Norm}\label{sec:tukey_UB}

We shall prove in this section an $\widetilde{O}(1/\eps^2)$ upper bound for estimating a mollified version of the Tukey loss $1$-norm $\Phi(x)$ for a vector $x\in \R^n$.

\subsection{Mollification of Tukey Loss Function} Consider the classic ``bump'' function $\psi$ and the standard scaled version $\psi_t$, which are defined as
\[
\psi(x) = \begin{cases}
			C_{\psi}\exp(-\frac{1}{1-x^2}), & |x|<1;\\
			0, & \text{otherwise},
		\end{cases}	\qquad
		\psi_t(x) = \frac{1}{t}\psi\left(\frac{x}{t}\right),
\]
where $C_{\psi}$ is the normalization constant such that $\int_{-1}^1 \psi(x)dx =1$. The following is a result on the decay of its Fourier transform. We define the Fourier transform of a function $f\in L^1(\R)$ to be $\hat f(t) = \int_{\R} e^{itx}f(x)dx$, and thus $f(x) = (2\pi)^{-1}\int_{\R} e^{-itx}\hat f(t)dt$ if $\hat f\in L^1(\R)$.

\begin{lemma}[{\cite{johnson:bump}}]\label{lem:bump_fourier}
There exists an absolute constant $C > 0$ such that $|\widehat{\psi}(t)| \leq C|t|^{-3/4}e^{-\sqrt{|t|}}$.
\end{lemma}

This enables us to upper bound the derivatives of $\psi$. This is probably a classical result but we do not know an appropriate reference and so we reproduce the proof here.

\begin{lemma}\label{lem:bump_derivative}
There exist absolute constants $C_1,C_2>0$ such that $\left\|\psi^{(k)}\right\|_\infty \leq C_1 (C_2k\log k)^{2k+2}$.
\end{lemma}
\begin{proof}
First note that
\[
\left\|\psi^{(k)}\right\|_\infty  \leq \frac{1}{2\pi}\left\|\widehat{\psi^{(k)}}\right\|_1 = \frac{1}{2\pi}\left\|(it)^k\widehat{\psi}(t)\right\|_1 = \frac{1}{2\pi}\int_\R |t|^k \left|\widehat\psi(t)\right| dt.
\]
On the one hand, taking $T = C_1(k\log k)^2$, it follows from Lemma~\ref{lem:bump_fourier} that
\[
|t|^k |\widehat\psi(t)| \leq C_2 t^{-5/4}, \qquad |t| > T,
\]
where $C_1, C_2 > 0$ are absolute constants. On the other hand,  $|\hat \psi(t)| \leq \|\psi\|_1 = 1$ for all $t$. Therefore
\begin{align*}
\int_\R |t|^k |\widehat\psi(t)| dt &\leq \int_{|t| > T} \frac{C_2}{t^{5/4}} dt + \int_{-T}^T |t|^k dt  \\
&\leq C_3 + 2T^{k+1}\\
&\leq C_4 (\sqrt{C_1} k\log k)^{2k+2}.
\end{align*}
The result then follows.
\end{proof}

Recall that the Tukey $1$-loss function is
\[
	\phi(x) = \begin{cases}
				|x|, &|x|\leq \tau;\\
				\tau, & |x| > \tau.
			\end{cases}
\]
It is easy to verify that $(\phi\ast \psi_{\tau/4})(x) = \phi(x)$ when $\tau/4\leq |x|\leq 3\tau/4$, and thus if we define
\[
	\widetilde\phi(x) = \begin{cases}
				\phi(x), & |x| \leq 3\tau/4; \\
				(\phi\ast\psi_{\tau/4})(x),  & |x| > 3\tau/4,
			\end{cases}
\]
it would be clear that $\widetilde\phi$ is infinitely times differentiable on $(0,\infty)$. Also observe that $\widetilde\phi(x) = \phi(x) = \tau$ when $|x| \geq 5\tau/4$, we see that $\widetilde\phi$ just mollifies $\phi$ on $[3\tau/4,5\tau/4]$. We shall take $\widetilde\phi$ to be our mollified version of the Tukey $1$-loss function.

Next we bound the derivatives of $\widetilde\phi$.

\begin{lemma}\label{lem:derivative_bounds}
There exist absolute constants $C_1,C_2>0$ such that $\left|\widetilde\phi^{(k)}(x)\right| \leq C_1 \tau (C_2k^2\log^2 k/\tau)^{k+1}$ for $x\in [3\tau/4,5\tau/4]$.
\end{lemma}
\begin{proof}
Observe that 
\[
\left|\widetilde\phi^{(k)}(x)\right| \leq \max_{x\in [\frac34\tau,\frac54\tau]} \phi(x) \cdot \left\|\psi_{\tau/4}^{(k)}\right\|_\infty \leq \tau\left\|\psi_{\tau/4}^{(k)}\right\|_\infty 
= \frac{\tau}{(\frac{\tau}{4})^{k+1}} \left\|\psi^{(k)}\right\|_\infty.
\]
The result follows from Lemma~\ref{lem:bump_derivative}.
\end{proof}

As a corollary of the preceding proposition, we have

\begin{lemma}\label{lem:mollified_approx} Let $a\in (0,\frac{3}{4})$ and $b>1$ be constants. There exists a polynomial $p(x)$ of degree $O(\frac{b-a}{\tau} \log^2 \frac{1}{\eps a\tau}\log^2\log \frac{1}{\eps a\tau})$ such that $|p(x)-\widetilde\phi(x)|\leq \eps\widetilde\phi(x)$ on $[a\tau,b\tau]$.
\end{lemma}
\begin{proof}
	Since $\widetilde\phi(x) \geq a\tau$ on $[a\tau,b\tau]$, it is sufficient to consider the uniform approximation $|p(x)-\widetilde\phi(x)|\leq \eps(a\tau)$ on $[a\tau,b\tau]$. It follows from Lemma~\ref{lem:poly_approx} that when $n>k$, 
	\[
	E_n(\widetilde\phi; [a\tau,b\tau]) \leq \frac{6^{k+1} e^k}{(k+1)n^k} \cdot \left(\frac{(b-a)\tau}{2}\right)^k\|\widetilde\phi^{(k)}\|_\infty \cdot \frac{1}{n-k}.
	\]
	Invoking Lemma~\ref{lem:derivative_bounds}, we obtain for $n\geq 2k$ that
	\[
	E_n(\widetilde\phi; [a\tau,b\tau])\leq \frac{C_1}{(k+1)(b-a)} \left(\frac{C_2(b-a)k^2\log^2 k}{\tau n}\right)^{k+1},
	\]
	where $C_1,C_2>0$ are absolute constants. It is now clear that we can take $k = O(\log\frac{1}{\eps a\tau})$ and $n = O(\frac{b-a}{\tau}k^2\log^2 k)$ so that $E_n(\widetilde\phi; [a\tau,b\tau])\leq \eps\cdot a\tau$.
\end{proof}

\subsection{Estimation Algorithm} Since $\widetilde{\phi}(x)$ agrees with $|x|$ for small $|x|$, it follows from Theorem~\ref{thm:power_function} that solving the subspace sketch problem for $\widetilde{\Phi}(x)$ requires $\widetilde{\Omega}(d/\eps^2)$ bits. In this subsection we show that this lower bound is tight up to polylogarithmic factors. Specifically we have the following theorem.

\begin{theorem}
Let $\widetilde\Phi(x)$ be the mollified Tukey loss $1$-norm of $x\in \R^n$. There exists a randomized algorithm which returns an estimate $Z$ to $\widetilde\Phi(x)$ such that $(1-\eps)\widetilde\Phi(x)\leq Z\leq (1+\eps)\widetilde\Phi(x)$ with probability at least $0.9$. The algorithm uses $\widetilde{O}(1/\eps^2)$ bits of space.
\end{theorem}

This theorem implies an $\widetilde{O}(d/\eps^2)$ upper bound for the corresponding subspace sketch problem. The remaining of the section is devoted to the proof of this theorem.

We shall first sample rows of $A$ with sampling rate $\Theta\left(\frac{\tau}{\widetilde{\Phi}(x) \eps^2}\right)$.
However, we do not know $\widetilde{\Phi}(x)$ in advance. To implement this, we sample rows of $A$ using $O(\log n)$ different sampling rates $1, (1.1)^{-1}, (1.1)^{-2}, \ldots, 1.1^{-O(\log n)}$, and in parallel, estimate $\widetilde{\Phi}(x)$ using a separate data structure of $O(\polylog(n) \cdot d)$ space~\cite{braverman2010zero, braverman2016streaming}, which gives an estimate $F$ satisfying $0.9 \widetilde{\Phi}(x) \le F \le 1.1 \widetilde{\Phi}(x)$.
Then we choose a sampling rate $r = 1.1^{-s}$ for some integer $s$ that is closest to $\frac{\tau}{F \eps^2}$. 
Thus $r\in \left[\frac{\tau}{2\Phi(x) \eps^2}, \frac{2\tau}{\widetilde{\Phi}(x) \eps^2}\right]$ when $\widetilde{\Phi}(x) > \frac{\tau}{2\varepsilon^2}$, and $r=1$ otherwise. 

Now we show that for the chosen sampling rate $r$, the sampled entries give an accurate estimation to $\widetilde{\Phi}(x)$.
This is definitely true when $r = 1$, in which case there is no sampling at all. 
Otherwise, let $X_i = \widetilde{\Phi}(x_i)$ if item $i$ is sampled and $X_i = 0$ otherwise. Let $X = \sum_i X_i$ and $Z = (1/r)X$. It is clear that $\E[ Z] = \widetilde{\Phi}(x)$. We calculate the variance below.
\begin{align*}
\Var(Z) = \frac{1}{r^2}\Var(X) = \frac{1}{r^2}\sum_i \Var(X_i)^2 
&= \frac{1}{r^2}\sum_i (r-r^2)(\widetilde{\Phi}(x_i))^2 \\
&\leq \frac{1}{r}\sum_i (\widetilde{\Phi}(x_i))^2 \\
&= O\left( \frac{\widetilde{\Phi}(x) \eps^2}{\tau} \cdot \sum_i \tilde\phi(x_i)\cdot \tau \right) \\
&= O(\eps^2)\cdot (\widetilde{\Phi}(x))^2.
\end{align*}
It follows from Chebyshev's inequality that with constant probability, 
$$
Z = \frac{1}{r} \sum X_i = (1 \pm O(\eps)) \widetilde{\Phi}(x).
$$
We condition on this event in the rest of the proof. Thus, it suffices to estimate the summation of $\widetilde{\Phi}(x_i)$ for those $x_i$ that are sampled. In the rest of this section, we use $L$ to denote the indices of entries that are sampled at the sampling rate $r$.

For each $i \in L$ with $|x_i| \ge \tau$, we claim that 
\begin{equation}\label{eqn:HH}
|x_i| \ge \Omega(\eps^2) \cdot \|(x_L)_{-O(1/\eps^2)}\|_1,
\end{equation}
where $x_L$ denotes the vector $x$ restricted to the indices in $L$ and $v_{-k}$ denotes the vector $v$ after zeroing out the largest $k$ entries in magnitude. 

We first show that $\widetilde{\Phi}(x_L) = O\left(\frac{\tau}{\eps^2} \right)$, 
which is clearly true when $r = 1$, since in this case, $\widetilde{\Phi}(x_L) = \widetilde{\Phi}(x) = O\left(\frac{\tau}{\eps^2} \right)$. When $r < 1$, $\sum_{i \in L} \widetilde{\Phi}(x_i) = (1 \pm O(\eps)) \cdot r \cdot \widetilde{\Phi}(x) =  O\left(\frac{\tau}{\eps^2} \right)$.

Let $L' = \{i\in L: |x_i|\geq\tau/2\}$. It follows that $|L'| \leq \widetilde{\Phi}(x_L)/(\tau/2) = O(1/\eps^2)$. Hence
\[
\|x_{L\setminus L'}\|_1 = \sum_{i \in L\setminus L'} |x_i| = \widetilde{\Phi}(x_{L\setminus L'}) \leq \widetilde{\Phi}(x_L) = 
O\left(\frac{\tau}{\eps^2} \right),
\]
establishing~\eqref{eqn:HH}.

Therefore, to find all $i \in L$ with $|x_i| \ge \tau$, we use an $\ell_p$-heavy hitter data structure, which can be realized by a \textsc{Count-Sketch} structure \cite{charikar2004finding} which hashes $x_L$ into $O(1/\beta)$ buckets and finds $\beta$-heavy hitters relative to $\|(x_L)_{-1/\beta}\|_1$. Set $\beta = \Theta(\eps^2)$. In the end we obtain a list $H\subseteq L$ such that every $i\in H$ is a $\beta/2$-heavy hitter relative to $\|(x_L)_{-1/\beta}\|_1$, and all $\beta$-heavy hitters are in $H$. 
Furthermore, for each $i\in H$ the data structure also returns an estimate $\hat{x}_i$ such that $|x_i| / 2 \le |\hat{x}_i| \le 2|x_i|$ whenever $|x_i|\geq\tau/2$.
The data structure has space complexity $\widetilde{O}(1/\eps^2)$. 

For each $x_i \in L$ with $|x_i| \ge \tau$, it must hold that $i \in H$. Let $H_1 = \{i\in H: |\hat{x}_i|\geq 5\tau/4\}$ and $H_2 = \{i\in H: 3\tau/8 \leq |\hat{x}_i| \leq 5\tau/2\}$.

For each $i\in H_1$, by the estimation guarantee it must hold that $|x_i|\geq 5\tau/4$. Hence $S_1 = \tau |H_1| = \widetilde{\Phi}(x_{H_1})$. 

For each $i\in H_2$ it must hold that $|x_i|\in [\frac{3}{16}\tau,5\tau]$, and thus 
\[
\|x_{[n]\setminus H_1}\|_1\leq 5\widetilde{\Phi}(x_{[n]\setminus H_1}).
\]
Let $p(x)$ be a polynomial such that $|p(x)-\widetilde{\phi}(x)| \leq \eps\widetilde{\phi}(x)$ on $[\frac{3}{16}\tau,5\tau]$. By Lemma~\ref{lem:mollified_approx}, it is possible to achieve $\deg p = O(\log^3(1/(\eps\tau)))$. We now use an estimation algorithm analogous to the \textsc{HighEnd} structure in~\cite{KNPW11}, which uses the same space $\widetilde{O}(1/\eps^2)$. Using the same \textsc{BasicHighEnd} structure in~\cite{KNPW11}, with constant probability, for each $x_i\in H$ we have $T$ estimates $\hat{x}_{i,1},\dots,\hat{x}_{i,T}\in \C$ such that $\hat{x}_{i,t} = x_i + \delta_{i,t}$, where each $\delta_{i,t}\in \C$ satisfies $|\delta_{i,t}|\leq |x_i|/2$ and $\E(\delta_{i,t})^k = 0$ for $k=1,\dots,3\deg p$. The estimator is 
\[
S_2 = \Re \sum_{i\in H_2} \widetilde{\Phi}\left(\frac{1}{T}\sum_{i=1}^T p(\hat{x}_{i,t})\right).
\]
It follows from the analysis in~\cite{KNPW11} that ($x$ can be replaced with $x_{[n]\setminus H_1}$ in the analysis of the variance) that the algorithm will output, with a constant probability,
\[
S_2 = \widetilde{\Phi}(x_{H_2}) \pm \eps\|x_{[n]\setminus H_1}\|_1 = (1\pm 5\eps)\widetilde{\Phi}(x_{H_2}).
\]

For each $i\in [n]\setminus(H_1\cup H_2)$, it must hold that $|x_i|\leq \tau/2$ and thus we can use an $\ell_1$ sketch algorithm as in~\cite{KNPW11}, and obtain $S_3 = (1\pm \eps)\widetilde{\Phi}(x_{[n]\setminus(H_1\cup H_2)})$. 

Finally, the algorithm returns $S_1+S_2+S_3$, which is a $(1\pm 10\eps)$-approximation to $\widetilde{\Phi}(x)$. Rescaling $\eps$ proves the correctness of the estimate.

For the part of evaluating $S_2$ and $S_3$, the space complexity is the same as the  \textsc{HighEnd} and $\ell_p$ sketch algorithm in~\cite{KNPW11}, which are both $\widetilde{O}(1/\eps^2)$ bits.

\section{An Upper Bound for $\ell_1$ Subspace Sketches in Two Dimensions}\label{sec:2dub}
In this section, we prove an $O(\polylog(n)/\eps)$ upper bound for the $\ell_1$ subspace sketch problem when $d = 2$.
Our plan is to reduce the $\ell_1$ subspace sketch problem with $d=2$ to coresets for the weighted $1$-median problem with $d = 1$.
For the latter problem, an $O(\polylog(n) / \varepsilon)$ upper bound is known \cite{har2007smaller}. 

For the special case where the first column of the $A$ matrix is all ones, the $\ell_1$ subspace sketch problem with $d=2$ is equivalent to coresets for $1$-median with $d = 1$. To see this, by homogeneity, we may assume $x_2 = 1$ for the query vector $x \in \mathbb{R}^2$.
Thus, $\|Ax\|_1 = \sum_{i = 1}^n |x_1 + A_{i, 2}|$, which is the $1$-median cost of using $x_2$ as the center on $\{-A_{1, 2}, -A_{2, 2}, \ldots, -A_{n, 2}\}$.
When entries of the first column of $A$ are positive but not necessarily all ones, we have
$$
\|Ax\|_1 = \sum_{i = 1}^n A_{i, 1}\left|x_1 + \frac{A_{i, 2} }{ A_{i, 1}}\right|,
$$
which is the weighted $1$-median cost of using $x_1$ as the center on $\{-A_{1, 2} / A_{1, 1}, -A_{2, 2} / A_{2, 1}, \ldots, -A_{n, 2} / A_{n, 1}\}$, with weights $\{A_{i, 1}, A_{i, 2}, \ldots, A_{n, 2}\}$.
It has been shown in \cite[Theorem~2.8]{har2007smaller} that there exists a coreset of size $O(\polylog(n) / \varepsilon)$ for the weighted $1$-median problem when $d = 1$.

For general $A$, we divide the rows of $A$ into three separate matrices $A^+$, $A^-$ and $A^{0}$.
Here, all entries in the first column of $A^+$ are positive, all entries in the first column of $A^-$ are negative, and all entries in the first column of $A^0$ are zeroes.
Since $\|Ax\|_1 = \|A^+x\|_1 + \|A^-x\|_1 + \|A^{0}x\|_1$, we can design subspace sketches separately for $A^+$, $A^-$ and $A^{0}$.
Our reduction above implies an $O(\polylog(n) / \varepsilon)$ upper bound for $A^+$ and $A^-$.
For $A^0$, since all entries in the first column are all zero, we have
$$
\|A^0x\|_1 = |x_2| \sum_{i} |A^0_{i, 2}|.
$$
Thus, it suffices to store $\sum_{i} |A^0_{i, 2}|$ for $A^0$.

\begin{theorem}
The $\ell_1$ subspace sketch problem can be solved using $\widetilde{O}(1 / \varepsilon)$ bits when $d = 2$.
\end{theorem}

\bibliography{literature}

\end{document}